\documentclass[11pt]{scrartcl}

\usepackage{url}
\usepackage[hidelinks]{hyperref}
\usepackage[utf8]{inputenc}
\usepackage[small]{caption}
\usepackage{graphicx}
\usepackage{amsmath}
\usepackage{booktabs}
\usepackage{amsthm}
\usepackage{amssymb}
\usepackage{pifont}
\usepackage{cleveref,bm,multirow}
\usepackage{color}
\usepackage{xspace}
\usepackage{tikz}
\usepackage{tabularx}
\usepackage{pbox}
\usepackage{framed}
\usepackage[shortlabels]{enumitem}
\usepackage{thm-restate}
\usepackage{nicematrix}
\usepackage{algorithm}
\usepackage{algorithmic}
\usepackage[most]{tcolorbox}
\usepackage{todonotes}

\usepackage{natbib}
\usepackage{babel}
\usepackage{subcaption}
\usetikzlibrary{arrows.meta}
\usetikzlibrary{positioning}
\renewcommand{\emptyset}{\varnothing}

\usepackage{cleveref}

\renewcommand{\emptyset}{\varnothing}

\newtheorem{theorem}{Theorem}[section]

\newtheorem{proposition}[theorem]{Proposition}
\newtheorem{lemma}[theorem]{Lemma}

\theoremstyle{definition}
\newtheorem{definition}[theorem]{Definition}
\newtheorem{example}[theorem]{Example}

\DeclareMathOperator*{\argmax}{arg\,max}

\newcommand{\mnwtie}{MNW$^\text{tie}$}

\newcommand{\therule}{\mathcal{F}}
\newcommand{\profile}{\mathcal{P}}
\newcommand{\instance}{\mathcal{I}}
\newcommand{\allocation}{\mathcal{A}}
\newcommand{\mnw}{MNW}

\newcommand*\circled[1]{\tikz[baseline=(char.base)]{
            \node[shape=circle,draw,inner sep=2pt] (char) {#1};}}

\newcommand\LL[1]{\multicolumn{1}{|c}{#1}}
\newcommand\RR[1]{\multicolumn{1}{c|}{#1}}

\newcommand*{\innerproofname}{Proof}
\newenvironment{innerproof}[1][\innerproofname]{\begin{proof}[#1]}{\end{proof}}

\allowdisplaybreaks

\title{Fair Division with Binary Valuations: Characterizations}

\author{
Florian Brandl\\University of Bonn
\and
Warut Suksompong\\National University of Singapore
\and
Nicholas Teh\\University of Oxford
}

\date{\vspace{-5ex}}

\begin{document}

\maketitle

\begin{abstract}
We consider the fair allocation of indivisible goods with binary valuations.
In this setting, the maximum Nash welfare rule, the leximin rule, and all additive welfarist rules with a strictly concave function coincide.
We show that for any number of agents, this rule is the only rule that satisfies envy-freeness up to one good, strategyproofness, neutrality, minimal completeness, and invariance under disapproving unassigned goods (IDU).
Moreover, we present an alternative characterization for two agents, where we replace IDU with non-redundancy and resource-monotonicity.
In both characterizations, all axioms are necessary.
\end{abstract}

\section{Introduction}
\label{sec:intro}

The fair allocation of resources is a central problem at the intersection of economics and computer science, with applications ranging from divorce settlement to inheritance division to load balancing \citep{BramsTa96,RobertsonWe98,Moulin03}.
A significant portion of recent work in the area investigates fairness considerations when allocating \emph{indivisible goods} such as jewelry, artwork, electronics, toys, and furniture \citep{GoldmanPr14,AmanatidisAzBi23}.

A key fairness desideratum is \emph{envy-freeness}, which states that each agent should value her own bundle of goods at least as much as every other agent's bundle. 
However, the trivial instance of two agents and a single valuable good shows that envy-freeness cannot always be attained (if this good must be allocated), so a relaxation is necessary.
One of the most intuitive relaxations is \emph{envy-freeness up to one good (EF1)}.
An allocation is said to be EF1 if any envy that an agent has towards another agent can be eliminated by removing a single good from the latter agent's bundle.
In addition to fairness, one may wish to achieve some form of efficiency.
A commonly studied efficiency property is \emph{Pareto-optimality (PO)}, which stipulates that no other allocation makes at least one agent better off without making another agent worse off.

A popular method for simultaneously achieving these two properties is the \emph{maximum Nash welfare (MNW)} rule, which selects an allocation maximizing the product of the agents' values. 
In a far-reaching paper, \citet{CaragiannisKuMo19} showed that under additive valuations, an MNW allocation always satisfies EF1 and PO.
Subsequently, \citet{YuenSu23} characterized MNW as the unique rule within the class of \emph{welfarist rules}---rules that choose an allocation maximizing some increasing function of the agents' values---that ensures EF1 under additive valuations.
This extends an earlier result by \citet{Suksompong23}, which provided this characterization only within the class of \emph{additive welfarist rules}, where the function can be expressed as a sum of some function of each agent's value (for MNW, this latter function is the logarithm function).

Despite these advantages, MNW has its drawbacks too. 
For example, it fails to be \emph{strategyproof}, which implies that an agent can sometimes benefit from misreporting her true preferences \citep{KlausMi02,HalpernPrPs20}.
It also violates \emph{resource-monotonicity}, so an agent could become worse off when an extra good is added to the pool \citep{ChakrabortyScSu21}.
Furthermore, finding or even approximating the maximum Nash welfare within a certain constant factor is a computationally hard problem \citep{lee2017apx}.
Note that many rules other than MNW satisfy EF1 and PO.
In particular, since EF1 and PO are properties of specific instances and do not connect different instances, we may consider an arbitrary instance and choose an EF1 and PO allocation different from the allocation chosen by MNW.
It thus appears to be a challenging task to identify a set of axioms that MNW satisfies but no other rule does.

An important subclass of additive valuations is the class of \emph{binary} valuations (sometimes referred to as \emph{binary additive} valuations), where each agent's value for each good is either $0$ or $1$.
Binary valuations can be viewed as approval votes, which have long been
studied in the voting literature \citep{brams2007approvalvoting,kilgour2010approval}, and permit very simple elicitation.
These valuations have therefore been investigated in
numerous fair division papers \citep{aleksandrov2015onlinefoodbank,darmann2015nashproduct,bouveret2016conflict,barman2018greedymnw,freeman2019equitable,HalpernPrPs20,hosseini2020infowithholding,kyropoulou2019groupallocation,amanatidis2021mnwefx,SuksompongTe22}.\footnote{Binary preferences have also been studied in matching \citep{BogomolnaiaMo04} and auctions \citep{MalikMi21}.}
In this subclass, the aforementioned drawbacks of MNW disappear: the rule is strategyproof \citep{HalpernPrPs20}, 
resource-monotone \citep{SuksompongTe22}, and can be computed in polynomial time \citep{darmann2015nashproduct,barman2018greedymnw}.

Before proceeding further, we highlight that under binary valuations, MNW coincides with the \emph{leximin} rule---which chooses an allocation that maximizes the smallest utility, then the second smallest utility, and so on---as well as any additive welfarist rule with a strictly concave function \citep{BenabbouChIg21}.\footnote{This includes the \emph{$p$-mean rule} for any (possibly negative) $p < 1$ \citep[Prop.~4.9]{CelineSuYu26}. 
\citet{EckartPsVe24} showed that for two agents with valuations that are additive but possibly non-binary, an allocation that maximizes the $p$-mean of normalized values for any $p\le 0$ satisfies EF1 and PO.
} Thus, our study should not be understood as distinguishing MNW from these equivalent rules, or as providing an argument in favor of MNW over leximin in the binary domain. 
Rather, for ease of exposition and consistency with existing work \citep{darmann2015nashproduct,barman2018greedymnw,HalpernPrPs20,amanatidis2021mnwefx}, we use the term ``MNW'' to refer to this common rule. 
While this convention gives a compact and familiar description of the rule, whenever we refer to MNW in the binary valuation domain, the corresponding statements apply equally to leximin as well as the equivalent strictly concave additive welfarist formulations.

In light of the advantages of MNW under binary valuations, one may be tempted to think that it is the indisputable choice in this domain.
However, it remains conceivable that there exist other rules which satisfy the same attractive properties (and perhaps more).
To see why the choice of rule is not obvious even under binary valuations, consider the following example.
\begin{example}
Consider two agents $N=\{1,2\}$ and eight goods $G = \{g_1,\dots,g_8\}$, where agent~$1$ approves $g_1,\dots,g_4$ and agent~$2$ approves $g_3,\dots,g_8$.
MNW gives $g_1,\dots,g_4$ to agent~$1$ and $g_5,\dots,g_8$ to agent~$2$.
However, the allocation that gives $g_1,g_2$ to agent~$1$ and $g_3,\dots,g_8$ to agent~$2$ is also envy-free (and therefore EF1) and PO.\footnote{Note that under binary valuations, every allocation satisfying PO maximizes utilitarian welfare (i.e., the sum of all agents' values), since such an allocation must assign each good to an agent who values it.
Hence, utilitarian welfare maximization cannot be used to distinguish among rules that ensure PO.}
Moreover, since properties like strategyproofness and resource-monotonicity connect different instances (i.e., each application of one of these properties involves two instances), they cannot be used to rule out the latter allocation in this instance, at least by reasoning about this instance alone.
\end{example}
Thus, in order to solidify the argument that MNW---or equivalently, the leximin rule and any additive welfarist rule with a strictly concave function---is a particularly compelling rule to use, it would be helpful to characterize it via sets of desirable properties.

\subsection{Our Contributions}

Throughout this paper, we assume that agents have binary valuations. 
In Section 3, we present a characterization of MNW with respect to five axioms: EF1, strategyproofness, neutrality, minimal completeness, and invariance under disapproving unassigned goods (IDU).
Neutrality is a basic property that states that the agents' received values should not change if the goods are simply relabeled.
Minimal completeness stipulates that all goods positively valued by at least one agent should be allocated and the remaining goods unallocated; a minimally complete version of MNW has been shown to satisfy attractive properties \citep{HalpernPrPs20,SuksompongTe22}.\footnote{On the other hand, if all goods must be allocated, no version of MNW can be strategyproof \citep{HalpernPrPs20}.}
IDU states that if an agent no longer approves a good that is not allocated to her, the returned allocation should remain the same.
Our characterization holds for any number of agents. 
We also show that all five properties are necessary---omitting any of them leads to a rule different from MNW that satisfies the remaining four properties.
Furthermore, we give a tie-breaking refinement of our characterization for the case of two agents by showing that, in addition to maximizing Nash welfare, the rule must break ties in a consistent manner across all instances with the same total number of goods and the same number of goods approved by at least one agent in this case.

Next, in Section 4, we provide an alternative characterization for the special case of two agents, by replacing IDU with non-redundancy (i.e., every allocated good should be given to an agent who values it) and resource-monotonicity.
In other words, we show that any rule satisfying EF1, strategyproofness, neutrality, minimal completeness, non-redundancy, and resource-monotonicity must maximize Nash welfare.
Again, we establish the independence among axioms as well as a tie-breaking refinement of our characterization.

To the best of our knowledge, these are the first characterizations of MNW within the space of all rules for any domain.
Together with prior work on the guarantees afforded by MNW under binary valuations, our results help cement MNW---or equivalently, the leximin rule and any additive welfarist rule with a strictly concave function---as the suitable allocation rule for this domain.

\section{Preliminaries}
\label{sec:prelims}

Let $N = [n]$ be the set of $n$ \emph{agents} and $G = \{g_1,g_2,\dots, g_m\}$ be the set of $m$ \emph{goods}, where $[k] := \{1,2,\dots,k\}$ for any positive integer $k$.
In \Cref{sec:alternative-characterization}, $G$ is not necessarily fixed, as extra goods may be added.
Each agent~$i \in N$ has a nonnegative \emph{valuation function}
$v_i$ over subsets of goods.
Every subset of goods in $G$ can be called a \emph{bundle}.

We assume that $v_i$ is \emph{binary additive}, meaning that $v_i(\{g\}) \in \{0,1\}$ for each $g\in G$ and $v_i(S) = \sum_{g \in S} v_i(\{g\})$ for any bundle $S\subseteq G$.
For simplicity, we sometimes write $v_i(g)$ instead of $v_i(\{g\})$ for $g \in G$.
We say that a good $g$ is \emph{unvalued} if $v_i(g) = 0$ for all agents $i \in N$, and \emph{valued} otherwise.
The list $\mathbf{v} = (v_1,\dots,v_n)$ is called a \emph{valuation profile}.
We define an \emph{approval profile} as the tuple $\profile = (P_1,\dots, P_n)$ such that for each $i\in N$, we have $g\in P_i$ if and only if $v_i(g) = 1$.
An \emph{instance} $\mathcal{I}=(N,G,\profile)$ is defined by the set of agents $N$, the set of goods $G$, and the approval profile $\profile$.

An allocation $\allocation = (A_1,\dots, A_n)$ is a list of disjoint bundles, where agent~$i \in N$ receives bundle $A_i$; let $\Pi_n(G)$ denote the set of all possible
allocations.
It is not necessary that $\bigcup_{i\in N}A_i = G$.
For each $n$, an \emph{allocation rule} for $n$ agents is a function $\therule$ that maps each instance with $n$ agents (and any set of goods) to an allocation.

A commonly studied fairness property is \emph{envy-freeness}, which states that each agent should value her assigned bundle at least as much as any other agent's bundle, that is, $v_i(A_i) \ge v_i(A_j)$ for all $i,j\in N$.
Since this is not always achievable when all (valued) goods must be allocated, a typical relaxation is \emph{envy-freeness up to one good (EF1)} \citep{Budish11,CaragiannisKuMo19}.\footnote{Under binary valuations, MNW is known to satisfy a stronger property called EFX$_0$ \citep{amanatidis2021mnwefx}. Our characterizations immediately carry over to EFX$_0$.}

\begin{definition}[Envy-freeness up to one good]
    An allocation $\allocation$ is \emph{envy-free up to one good (EF1)} if for all $i,j \in N$ with $A_j\neq\emptyset$, there exists a good $g \in A_j$ such that $v_i(A_i) \geq v_i(A_j \setminus \{g\})$.
\end{definition}

For efficiency, a basic notion is \emph{Pareto-optimality (PO)}.

\begin{definition}[Pareto-optimality]
    An allocation $\allocation$ is \emph{Pareto-optimal (PO)} if there does not exist another allocation $\allocation'$ such that $v_i(A'_i) \geq v_i(A_i)$ for all $i \in N$ and the inequality is strict for at least one $i \in N$.
\end{definition}

We also consider a related property called \emph{non-redundancy}, which states that each allocated good must be given to an agent who values it. 
This property has been used in several papers \citep{BabaioffEzFe21-dichotomous,BenabbouChIg21,BarmanVe22,SuksompongTe23,GokoIgKa24,MontanariScSu25,ViswanathanZi25}.\footnote{Some authors referred to this property as \emph{clean} or \emph{non-wasteful}.}

\begin{definition}[Non-redundancy]
    An allocation $\allocation$ is \emph{non-redundant} if $v_i(g) = 1$ for all $i \in N$ and $g \in A_i$.
\end{definition}

A slightly different property is minimal completeness \citep{HalpernPrPs20,SuksompongTe22}.

\begin{definition}[Minimal completeness]
An allocation $\allocation$ is \emph{minimally complete} if it allocates all valued goods and no unvalued goods.     
\end{definition}

One can check that while none of PO, non-redundancy, and minimal completeness implies another, the combination of any two implies the third.

We say that an allocation rule $\therule$ is EF1 (resp., PO, non-redundant, minimally complete) if for all instances $\mathcal{I}$, the allocation $\allocation = \therule(\mathcal{I})$ is EF1 (resp., PO, non-redundant, minimally complete).

Another important property of allocation rules is \emph{strategyproofness}, which stipulates that no agent should be able to misreport her preference so as to obtain a strictly higher value for herself. 
\begin{definition}[Strategyproofness]
    An allocation rule $\therule$ is \emph{strategyproof} if there do not exist valuation profiles $\mathbf{v}$ and $\mathbf{v}'$ with the following property: Denoting the instances corresponding to $\mathbf{v}$ and $\mathbf{v}'$ by $\mathcal{I}$ and $\mathcal{I}'$, respectively, there exists $i\in N$ such that
    \begin{itemize}
    \item $v_j = v'_j$ for all $j\in N\setminus\{i\}$;
    \item $v_i(A'_i) > v_i(A_i)$,
    \end{itemize}
    where $\allocation = \therule(\mathcal{I})$ and $\allocation' = \therule(\mathcal{I}')$.
\end{definition}

The next property, \emph{neutrality}, captures a different type of fairness in that the agents' received values should not depend on the labeling of the goods.

\begin{definition}[Neutrality]
    An allocation rule $\therule$ is \emph{neutral} if the following holds:    
    For any two instances $\mathcal{I}$ and $\mathcal{I}'$ such that $\mathcal{I}'$ can be obtained from $\mathcal{I}$ by permuting the labels of the goods, if $\therule(\mathcal{I}) = \allocation$ and $\therule(\mathcal{I}') = \allocation'$, then $v_i(A_i) = v'_i(A'_i)$ for each $i \in N$.
\end{definition}

Note that the companion property of \emph{anonymity}, which asserts that the agents' values should not depend on the labeling of the \emph{agents}, cannot be satisfied along with minimal completeness or PO. Indeed, if there are two agents and one good that is valued by both agents, then one agent must receive value~$1$ and the other agent~$0$.

We now introduce a new property that we call \emph{invariance under disapproving unassigned goods (IDU)}.
This property states that if an agent no longer approves a good that is not allocated to her, then the allocation returned by the rule should remain the same.
IDU is similar in spirit to \emph{independence of irrelevant alternatives}, which is frequently considered in social choice settings \citep{Arrow50}.
Specifically, if an allocation rule does not assign good~$g$ to agent~$i$, then agent~$i$ lowering her value for~$g$ is irrelevant for achieving the allocation that the rule deems best.
Thus, agent $i$ ceasing to approve $g$ is an irrelevant change to the valuations (as deemed by the rule), and therefore the allocation should not change.
IDU also has the flavor of \emph{monotonicity} (i.e., an agent should not be able to obtain a good by ceasing to value it) as well as \emph{non-bossiness} (i.e., no agent can change the allocation of another agent without changing her own allocation) \citep{Thomson16}.
Several natural rules, including all welfarist rules and picking sequences, satisfy IDU (provided that tie-breaking is consistent).\footnote{Even for general additive valuations, these rules satisfy a natural generalization of IDU where if an agent lowers her value for a good not allocated to her to~$0$, the returned allocation should not change.}

\begin{definition}[Invariance under disapproving unassigned goods]
    An allocation rule $\therule$ is \emph{invariant under disapproving unassigned goods (IDU)} if the following holds:     
    For any two instances $\mathcal{I} = (N,G,\profile)$ and $\mathcal{I}' = (N,G,\profile')$ such that $\mathcal{I}$ differs from $\mathcal{I}'$ in that $g\in P_i$ and $g\not\in P'_i$ for some good~$g$, if $\therule(\mathcal{I}) = \allocation$ and $g\not\in A_i$, then $\therule(\mathcal{I}') = \therule(\mathcal{I})$.
\end{definition}

We observe that IDU and minimal completeness together imply PO, a fact that will be useful later.

\begin{lemma}\label{lem:PO}
    Any allocation rule (for any $n$) satisfying IDU and minimal completeness also satisfies PO.
\end{lemma}

\begin{proof}
    Consider a rule $\therule$ satisfying IDU and minimal completeness, and take any instance.
    Minimal completeness implies that all valued goods are allocated and all unvalued goods are unallocated.
    Suppose for contradiction that some valued good $g$ is allocated to an agent~$i$ who does not value it.
    Consider a modified instance where all other agents do not value $g$.
    By IDU, $g$ should still be allocated to~$i$.
    However, this contradicts minimal completeness, as $g$ is now unvalued.
    Hence, every valued good must be allocated to an agent who values it, and PO readily follows.
\end{proof}

One can check that the converse of \Cref{lem:PO} does not hold---in fact, PO implies neither minimal completeness nor IDU---and minimal completeness or IDU alone does not imply PO.

We now define the main rule of interest in this paper.

\begin{definition}[Maximum Nash welfare]
\label{def:MNW}
    Given an instance, an allocation $\mathcal{A}$ is a \emph{maximum Nash welfare (MNW) allocation} if, among the set of allocations in $\Pi_n(G)$, it maximizes the number of agents receiving positive value and, subject to that, maximizes the  product of positive values. 
    Formally, let $H(\mathcal{A}) = \{i \in N : v_i(A_i) > 0 \}$ and $\mathcal{H} = \argmax_{\mathcal{A} \in \Pi_n(G)} |H(\mathcal{A})|$. For any instance $\instance$, let \mnw($\instance$) $= \argmax_{\mathcal{A'}\in \mathcal{H}} \prod_{i \in H(\mathcal{A'})} v_i(A'_i)$ denote the set of all MNW allocations in $\instance$, and call $\mathcal{A}$ an MNW allocation if $\mathcal{A} \in$ \mnw($\instance$).     

    We say that a rule \emph{maximizes Nash welfare} (equivalently, is an MNW rule) if, given any instance, it always returns an MNW allocation.
    If there is more than one MNW allocation, the rule is allowed to break ties arbitrarily.
\end{definition}

\citet{HalpernPrPs20} proposed MNW$^{\text{tie}}$, which breaks ties for MNW by discarding unvalued goods and maximizing the agents' values lexicographically.
If we further break any remaining ties according to a lexicographic order over the goods, MNW$^{\text{tie}}$ satisfies EF1, PO, non-redundancy, minimal completeness, strategyproofness, neutrality, and IDU.
For an illustrative example, consider the following instance:
        \begin{center}
            \begin{tabular}{ c | c c c c c c }
            $\profile$ & $g_1$ & $g_2$ & $g_3$ & $g_4$ & $g_5$ & $g_6$ \\ 
             \hline
             $1$ & \circled{$1$} & \circled{$1$} & $1$ & \circled{$1$} & $0$ & $0$ \\  
             $2$ & $1$ & $1$ & \circled{$1$} & $0$ & \circled{$1$} & $0$ \\
            \end{tabular}
        \end{center}
First, MNW$^{\text{tie}}$ discards $g_6$ because it is unvalued, and allocates $g_4$ (resp., $g_5$) to agent~$1$ (resp., agent~$2$) because it is valued by only one agent.
As $g_1,g_2,g_3$ are valued by both agents and each agent already receives value~$1$ from the remaining goods, MNW$^{\text{tie}}$ must allocate two of these three goods to one agent and one to the other agent.
By the lexicographic tie-breaking over agents' values, MNW$^{\text{tie}}$ allocates two of $g_1,g_2,g_3$ to agent~$1$ and one to agent~$2$.
To break the remaining ties according to a lexicographic order over the goods, the rule allocates $g_1,g_2$ (i.e., the lowest-index goods) to agent~$1$ and $g_3$ to agent~$2$.
This results in the circled allocation.

\section{Main Characterization}

This section details our main result, which characterizes MNW as the only rule satisfying EF1, strategyproofness, neutrality, minimal completeness, and IDU under binary valuations.
We show that this characterization holds for any number of agents $n \geq 2$.

We first state a lemma of \citet{HalpernPrPs20}, which provides necessary and sufficient conditions for an allocation to be MNW.\footnote{This is Lemma 5 in the extended version of their work.}
A \emph{path} is a sequence of agents together with an allocation such that each agent (except the first) values some good in her predecessor's bundle, and the path is ``critical'' if the first agent values her own bundle at least two more than the last agent values his own bundle.
More precisely, given an instance $\instance = (N,G,\profile)$ and an allocation $\allocation$, a path in the pair $(\instance, \allocation)$ is a sequence of agents $(i_1,\dots,i_k)$ such that $A_{i_\ell} \cap P_{i_{\ell + 1}} \neq \emptyset$ for all $\ell \in \{1,\dots,k-1\}$. 
A path $(i_1,\dots, i_k)$ in $(\instance,\allocation)$ is \emph{critical} if $v_{i_1}(A_{i_1}) > v_{i_k}(A_{i_k}) + 1$.
For example, consider the following instance along with the circled allocation~$\allocation$:
        \begin{center}
            \begin{tabular}{ c | c c c c c c c c c }
            $\profile$ & $g_1$ & $g_2$ & $g_3$ & $g_4$ & $g_5$ & $g_6$ & $g_7$ & $g_8$ & $g_9$ \\ 
             \hline
             $1$ & \circled{$1$} & \circled{$1$} & \circled{$1$} & $1$ & $1$ & $0$ & $0$ & $0$ & $0$ \\  
             $2$ & $0$ & $0$ & $1$ & \circled{$1$} & \circled{$1$} & \circled{$1$} & \circled{$1$} & $0$ & $0$ \\
             $3$ & $0$ & $1$ & $0$ & $0$ & $1$ & $1$ & $0$ & \circled{$1$} & $0$ \\
             $4$ & $0$ & $1$ & $1$ & $0$ & $1$ & $0$ & $1$ & $0$ & \circled{$1$} \\
            \end{tabular}
        \end{center}
Observe that $(1,2,3)$ is a path, since agent~$2$ (resp., agent~$3$) values some good in agent~$1$'s (resp., agent~$2$'s) bundle.
Moreover, this path is critical, as $v_1(A_1) = 3 > 1 + 1 = v_3(A_3) + 1$.
On the other hand, $(1,2,3,4)$ is not a path, since agent~$4$ does not value any good in agent~$3$'s bundle.

\begin{lemma}[\citet{HalpernPrPs20}] \label{lem:mnw_conditions}
    For any $(\instance, \allocation)$ such that $\allocation$ is PO, it holds that $\allocation \in $ \emph{\mnw}($\instance$) if and only if there is no critical path in $(\instance,\allocation)$.
\end{lemma}

Then, our characterization result is as follows.
\begin{theorem}
\label{thm:main-characterization}
    For any fixed $n \ge 2$, under binary valuations, any allocation rule that is \emph{EF1}, strategyproof, neutral, minimally complete, and \emph{IDU} maximizes Nash welfare.
\end{theorem}
\begin{proof}
    Let $\therule$ be a rule satisfying the five axioms in the theorem statement.
    Since $\therule$ is minimally complete and IDU, by \Cref{lem:PO}, we can assume that $\therule$ is PO.
    
    Suppose for contradiction that $\therule$ does not maximize Nash welfare.
    By Lemma \ref{lem:mnw_conditions} and the assumption that $\therule$ satisfies PO, there must exist an instance $\instance$ such that $(\instance, \therule(\instance))$ admits a critical path.
    By PO and minimal completeness, for an allocation $\allocation$ returned by $\therule$ in any instance, we have $v_i(A_i) = |A_i|$ for all $i \in N$.

    Let $k$ be the length of a shortest critical path across all instances, and let $\instance = (N,G,\profile)$ and $\allocation = \therule(\instance)$ be such an instance and its corresponding allocation under $\therule$.
    Then, we have that $(\instance, \allocation)$ admits a critical path $(i_1, \dots, i_k$) of length $k$. Without loss of generality, let $i_\ell = \ell$ for all $\ell \in [k]$.

    Since $(1,\dots, k)$ is a critical path, by definition of a critical path, it holds that
    \begin{equation} \label{eqn:A1>Ak+1}
        |A_1| > |A_k|+1.
    \end{equation}
    We split our analysis into two cases.

    \paragraph{Case 1: $k \geq 3$.}
    
    By minimal completeness and IDU, we may assume that 
    \begin{enumerate}[(i)]
        \item for all $i \in [k]$ and $j \in N \setminus [k]$, $P_i \cap A_j = \emptyset$ and $P_j \cap A_i = \emptyset$;
        \item for all $\ell \in [k]\setminus\{1\}$, $|P_\ell| = |A_\ell| + 1$ and $|P_\ell \cap A_{\ell-1}|=1$; and
        \item $P_1 = A_1$.
    \end{enumerate}
    Indeed, for (i) and (iii), we may let, e.g., each agent~$i\in [k]$ disapprove all goods in the bundle $A_j$ for every $j\in N\setminus [k]$; the allocation output by $\therule$ does not change due to IDU.
    Similarly, for (ii), we let each agent~$\ell\in[k]\setminus\{1\}$ cease to approve all goods in $A_{\ell-1}$ except one; the existence of the latter good follows from the definition of a path.
    For $\ell \in [k] \setminus \{1\}$, let $g_\ell$ be the unique good in $P_\ell \cap A_{\ell-1}$.

    If $|A_1| -1 > |A_{\ell}|$ for some $\ell \in [k-1]\setminus \{1\}$, then $(1,\dots,\ell)$ is a critical path shorter than $(1,\dots,k)$, a contradiction.
    On the other hand, if $|A_1|-1 < |A_\ell|$, then since $|A_1| > |A_k|+1$ (by definition of a critical path), we get $|A_\ell| \ge |A_1| > |A_k| + 1$, which means that $|A_\ell| - 1 > |A_k|$ and so $(\ell,\dots,k)$ is a critical path shorter than $(1,\dots,k)$, again a contradiction.
    Thus, we must have
    \begin{equation} \label{eqn:case1_A1-1=Al}
        |A_1| - 1 = |A_\ell| \quad \text{for all} \quad \ell \in [k-1] \setminus \{1\}.
    \end{equation}
    Now, if $|A_k| \geq |A_\ell|$ for some $\ell \in [k-1] \setminus \{1\}$, then we get $|A_1| - 1 = |A_\ell|  \leq |A_k|$, contradicting (\ref{eqn:A1>Ak+1}).
    Therefore, it must be that $|A_k| < |A_\ell|$, or equivalently,
    \begin{equation*}
        |A_k| + 1 \leq |A_\ell| \quad \text{for all} \quad \ell \in [k-1] \setminus \{1\}.
    \end{equation*}
    Then, if $|A_k| + 1 < |A_\ell|$ for some $\ell \in [k-1]\setminus \{1\}$, we again have that $(\ell, \dots, k)$ is a critical path shorter than $(1,\dots,k)$, a contradiction.
    Hence,
    \begin{equation*}
        |A_k| + 1 = |A_\ell| \quad \text{for all} \quad \ell \in [k-1] \setminus \{1\}.
    \end{equation*}    
    Combining this with (\ref{eqn:case1_A1-1=Al}), we get
    \begin{equation*} 
        |A_1| - 1 = |A_2| = \dots = |A_{k-1}| = |A_k| + 1 := r.
    \end{equation*}
    The valuations and corresponding bundles of agents in the critical path of $(\instance, \allocation)$ are illustrated in \Cref{tab:m=3_P}.

    \begin{table*}[t]
    \begin{center}
        \resizebox{\columnwidth}{!}{
        \begin{tabular}{ c | c c c c c c c c c c c c c c c c c | c }
        \multirow{2}{*}{ $\profile$
        }
         & \multicolumn{3}{c}{$\overbrace{\smash{}\hspace{2cm}\smash{}}^{A_1 = P_1}$} & \multicolumn{3}{c}{$\overbrace{\smash{}\hspace{2cm}\smash{}}^{A_2 = P_2 \setminus \{g_2\}}$} & & & \multicolumn{3}{c}{$\overbrace{\smash{}\hspace{2.8cm}\smash{}}^{A_{k-2} = P_{k-2} \setminus \{g_{k-2}\}}$} & \multicolumn{3}{c}{$\overbrace{\smash{}\hspace{2.8cm}\smash{}}^{A_{k-1} = P_{k-1} \setminus \{g_{k-1}\}}$} & \multicolumn{3}{c|}{$\overbrace{\smash{}\hspace{2.2cm}\smash{}}^{A_{k} = P_{k} \setminus \{g_k\}}$} &\multirow{1}{*}{bundle}
         \\
          &  & & \rotatebox{90}{$g_2$} & & & \rotatebox{90}{$g_3$} & $\cdots$ & \rotatebox{90}{$g_{k-2}$} & & & \rotatebox{90}{$g_{k-1}$} & & & \rotatebox{90}{$g_k$} & & & & size/value\\ 
         \hline
         $1$ & \LL{$1$} & $\cdots$ & \RR{$1$} & & & & & & & & & & & & & & & $r+1$ \\  
         \cline{2-4}\cline{5-7}
         $2$ & & & $1$ & \LL{$1$} & $\cdots$ & \RR{$1$} & & & & & & & & & & & & $r$\\
         \cline{5-7}
         \vdots & & & & & & & $\ddots$ & & & & & & & & & & & \vdots\\
          \cline{10-12}
         $k-2$ & & & & & & & & $1$ & \LL{$1$} & \hspace{2.5mm}$\cdots$ & \RR{$1$} & & & & & & & $r$\\
         \cline{10-12} \cline{13-15}
         $k-1$ & & & & & & & & & & & $1$ & \LL{$1$} & \hspace{2.5mm}$\cdots$ & \RR{$1$} & & & & $r$\\
         \cline{13-15} \cline{16-18}
         $k$ & & & & & & & & & & & & & & $1$ & \LL{$1$} & $\cdots$ & \RR{$1$} & $r-1$\\
         \cline{16-18}
        \end{tabular}
        }
        \caption{Profile $\profile$ for instance $\instance$ in Case~1 of the proof of \Cref{thm:main-characterization}. Empty cells indicate a value of $0$.\label{tab:m=3_P}}
    \end{center}
    \end{table*}

    We will show that $k$ cannot be the length of a shortest critical path across all profiles, by constructing a profile with a critical path of length strictly less than $k$, thereby arriving at a contradiction.

    Consider another instance $\instance' = (N, G, \profile')$ with profile $\profile' = (P'_1,\dots, P'_n)$ such that
    \begin{equation*}
        P'_{k-2} = P_{k-2} \cup \{g_k\} \quad \text{and} \quad P'_\ell = P_\ell \text{ for all } \ell \in N \setminus \{k-2\}.
    \end{equation*}
    Let $\allocation' = \therule(\instance')$ be the corresponding allocation returned by $\therule$ on this instance $\instance'$.        
    The valuations of agents $\{k-2,k-1,k\}$ in the critical path of $(\instance, \allocation)$ under the new instance $\instance'$ are illustrated as follows. Empty cells indicate that the agent has a value of $0$ for the good, and the valuations of agents $\{1,\dots,k-3\}$ remain the same as in $\profile$.
    \begin{center}
        \resizebox{\columnwidth}{!}{
        \begin{tabularx}{\linewidth}{ c | >{\centering\arraybackslash}X >{\centering\arraybackslash}X >{\centering\arraybackslash}X >{\centering\arraybackslash}X >{\centering\arraybackslash}X >{\centering\arraybackslash}X >{\centering\arraybackslash}X >{\centering\arraybackslash}X >{\centering\arraybackslash}X >{\centering\arraybackslash}X >{\centering\arraybackslash}X}
         $\profile'$ & $\cdots$ & \rotatebox{90}{$g_{k-2}$} & &  & \rotatebox{90}{$g_{k-1}$} & & & \rotatebox{90}{$g_k$} \\ 
         \hline
         \vdots & $\ddots$ & & & & & & & & & & \\  
         $k-2$ & & $1$ & $1$ & $\cdots$ & $1$ & & & $\mathbf{1}$ & & & \\  
         $k-1$ & & & & & $1$ & $1$ & $\cdots$ & $1$ & & & \\
         $k$ & & & & & & & & $1$ & $1$ & $\cdots$ & $1$ \\
        \end{tabularx}
        }
        \end{center}
        We consider two further cases, based on whether $g_k$ is contained in $A'_{k-2}$.

        \textbf{Case 1(a): $g_k \in A'_{k-2}$.}
        We know that $P'_k = P_k = A_k \cup \{g_k\}$ and agent~$k$ is the only agent who values goods in $A_k$ (by (i)).
        Hence, by PO, we have that $|A'_k| = |A_k| = r-1$.

        If $g_{k-2} \in A'_{k-2}$ (when $k \ge 4$) or $g_{k-1} \in A'_{k-2}$, then since 
        agent~$k-2$ is the only agent who values goods in $A_{k-2} \setminus \{g_{k-1}\}$, by PO, we have that $|A'_{k-2}| \geq |A_{k-2} \cup \{g_k\} \setminus \{g_{k-1}\}| + 1 \ge r+1$.
        This means that $(k-2, k)$ forms a critical path of length $2$, contradicting the assumption that $k \geq 3$ is the length of a shortest critical path across all profiles.
        Thus, we must have that $g_{k-2} \notin A'_{k-2}$ (when $k\ge 4$) and $g_{k-1} \notin A'_{k-2}$.
        If $k = 3$, then $(k-2, k)$ forms a critical path of length~$2$ regardless, so it must hold that $k\ge 4$.
        Then, we have $|A'_{k-2}| = r$.

        Recall that $|A_1| = r+1$ and $|A_\ell| = r$ for all $\ell \in \{2,\dots, k-3\}$.
        Since $g_{k-2}\not\in A'_{k-2}$, this means that $|A'_1 \cup \dots \cup A'_{k-3}| = r+1 + (k-4) \cdot r = (k-3) \cdot r + 1$.
        Thus, by the pigeonhole principle, there must exist an agent~$\ell \in \{1,\dots,k-3\}$ such that $|A'_\ell| = r+1$ (by PO, we cannot have $|A'_\ell| > r+1$).
        Let $\ell$ be the largest such index in $\{1,\dots,k-3\}$.
        
        If there exists an agent~$\ell' \in \{\ell+1,\dots, k-3\}$ with $|A'_{\ell'}| = r-1$ (by PO, we cannot have $|A'_{\ell'}| < r-1$), then let $\ell'$ be the smallest such index in $\{\ell+1,\dots, k-3\}$. 
        Consequently, there exists a critical path $(\ell, \ell+1,\dots, \ell')$ with length strictly less than $k$, contradicting the assumption that $k$ is the length of a shortest critical path across all profiles.
        Otherwise, all agents $\ell' \in \{\ell+1,\dots, k-3\}$ have $|A'_{\ell'}| = r$. Then, together with the fact that $|A'_{k-2}| = r$, there exists a critical path $(\ell, \ell+1, \dots, k-2, k)$ with length strictly less than $k$, again giving us a contradiction.
        
        \textbf{Case 1(b): $g_k \notin A'_{k-2}$.}
        IDU implies that 
        \begin{equation} \label{eqn:case1b_A=A'}
            \allocation = \allocation'.
        \end{equation}

        Consider another instance $\instance'' = (N, G, \profile'')$ with profile $\profile'' = (P''_1,\dots, P''_n)$ such that
        \begin{equation*}
            P''_k = P'_k \cup \{g_{k-1}\} \quad \text{and} \quad P''_\ell = P'_\ell \text{ for all } \ell \in N \setminus \{k\}.
        \end{equation*}
        Let $\allocation'' = \therule(\instance'')$ be the corresponding allocation returned by $\therule$ on this instance $\instance''$.
        The valuations of agents $\{k-2,k-1,k\}$ in the critical path of $(\instance, \allocation)$ under the new instance $\instance''$ are illustrated as follows. The valuations of agents $\{1,\dots,k-3\}$ remain the same as in $\profile'$ (and $\profile$).
        \begin{center}
        \resizebox{\columnwidth}{!}{
        \begin{tabularx}{\linewidth}{ c | >{\centering\arraybackslash}X >{\centering\arraybackslash}X >{\centering\arraybackslash}X >{\centering\arraybackslash}X >{\centering\arraybackslash}X >{\centering\arraybackslash}X >{\centering\arraybackslash}X >{\centering\arraybackslash}X >{\centering\arraybackslash}X >{\centering\arraybackslash}X >{\centering\arraybackslash}X}
         $\profile''$ & $\cdots$ & \rotatebox{90}{$g_{k-2}$} & & & \rotatebox{90}{$g_{k-1}$} & & & \rotatebox{90}{$g_k$} \\ 
         \hline
         \vdots & $\ddots$ & & & & & & & & & & \\  
         $k-2$ & & $1$ & $1$ & $\cdots$ & $1$ & & & 1 & & & \\  
         $k-1$ & & & & & $1$ & $1$ & $\cdots$ & $1$ & & & \\
         $k$ & & & & & $\mathbf{1}$ & & & $1$ & $1$ & $\cdots$ & $1$ \\
        \end{tabularx}
        }
        \end{center}        
        Note that under instance $\instance''$, $g_{k-1}$ and $g_k$ are approved by the same set of agents (namely, $\{k-2,k-1,k\}$).
        
        Strategyproofness for agent~$k$ at profile $\mathcal P'$ implies that
        \begin{equation} \label{eqn:case1b_sp_A''k}
            |A''_k \cap P'_k| \leq r-1;
        \end{equation} 
        otherwise agent~$k$ can manipulate from $P'_k$ to $P''_k$ and strictly benefit.
        Moreover, since agent~$k$ is the only agent who approves goods in $A_k = A'_k = P'_k \setminus \{g_k\}$, by PO and (\ref{eqn:case1b_sp_A''k}), $A''_k \cap P'_k = A'_k$ (i.e., $g_k \notin A''_k$).
        
        Suppose for contradiction that $A''_k \neq A'_k$.
        Then, by PO,
        \begin{equation} \label{eqn:case1b_A''k=A'k+gk-1}
            A''_k = A'_k \cup \{g_{k-1}\}.
        \end{equation}
        Let $\pi$ be the permutation of $G$ swapping $g_{k-1}$ and $g_k$ and fixing all other goods.
        Consider the instance $\instance''' = (N, G, \profile''')$ with profile $\profile''' = \pi(\profile')$, and let $\allocation''' = \therule(\instance''')$ be the corresponding allocation returned by $\therule$ on the instance $\instance'''$.
        The valuations of agents $\{k-2,k-1,k\}$ in the critical path of $(\instance, \allocation)$ under the new instance $\instance'''$ are illustrated as follows. The valuations of agents $\{1,\dots,k-3\}$ remain the same as in $\profile'$ (and $\profile$).
        \begin{center}
        \resizebox{\columnwidth}{!}{
        \begin{tabularx}{\linewidth}{ c | >{\centering\arraybackslash}X >{\centering\arraybackslash}X >{\centering\arraybackslash}X >{\centering\arraybackslash}X >{\centering\arraybackslash}X >{\centering\arraybackslash}X >{\centering\arraybackslash}X >{\centering\arraybackslash}X >{\centering\arraybackslash}X >{\centering\arraybackslash}X >{\centering\arraybackslash}X}
         $\profile'''$ & $\cdots$ & \rotatebox{90}{$g_{k-2}$} & & & \rotatebox{90}{$g_{k-1}$} & & & \rotatebox{90}{$g_k$} \\ 
         \hline
         \vdots & $\ddots$ & & & & & & & & & & \\  
         $k-2$ & & $1$ & $1$ & $\cdots$ & $1$ & & & $1$ & & & \\  
         $k-1$ & & & & & $1$ & $1$ & $\cdots$ & $1$ & & & \\
         $k$ & & & & & $1$ & & & & $1$ & $\cdots$ & $1$ \\
        \end{tabularx}
        }
        \end{center}
        Note that $\profile'''$ can be derived from $\profile''$ by replacing $P''_k$ with $P''_k \setminus \{g_k\}$.
        
        Now, we have that
        \begin{equation} \label{eqn:case1b:A''k=Ak+gk-1}
            A''_k = A'_k \cup \{g_{k-1}\} = A_k \cup \{g_{k-1}\},
        \end{equation}
        where the equalities follow from (\ref{eqn:case1b_A''k=A'k+gk-1}) and (\ref{eqn:case1b_A=A'}), respectively.
        Since $g_k \notin A_k$, it holds that $g_k \notin A''_k$.
        Then, by IDU between $\profile''$ and $\profile'''$, we have $\allocation''' = \allocation''$.
        Combining this with~(\ref{eqn:case1b:A''k=Ak+gk-1}), we get that 
        \begin{equation} \label{case1b_contradiction}
            A'''_k = A'_k \cup \{g_{k-1}\}.
        \end{equation}
        On the other hand, since $\profile''' = \pi(\profile')$, neutrality implies that $|A_k'''| = |A_k'|$, a contradiction with~(\ref{case1b_contradiction}).
        Hence, $A''_k = A'_k$.
        
        By IDU between $\profile''$ and $\profile'$, we have
        \begin{equation} \label{eqn:case1b_A''=A'}
            \allocation'' = \allocation'.
        \end{equation}
        Since $g_{k-1} \in A_{k-2}$ and $\allocation = \allocation''$ (by (\ref{eqn:case1b_A=A'}) and (\ref{eqn:case1b_A''=A'})), it follows that $g_{k-1} \in A''_{k-2}$.
        Moreover, since $g_{k-1} \in P''_k$, the path $(1,\dots,k-2,k)$ is a critical path of length $k-1$ in $(\instance'',\allocation'')$.
        This contradicts the assumption that $k$ is the length of a shortest critical path across all profiles.
    \paragraph{Case 2: $k =2$.}
        In this case, we first use minimal completeness and IDU to identify an instance such that agent~$1$ gets all her valued goods, agent~$2$ gets all his valued goods except one good which is assigned to agent~$1$, and agent~$1$ gets at least two more goods than agent~$2$.
        We then show by induction that as agent~$2$ approves additional goods in agent~$1$'s bundle, the size of agent~$2$'s bundle does not change; the induction step involves applying neutrality.
        When agent~$2$ approves $|A_2|+2$ goods in agent~$1$'s bundle, she receives only $|A_2|$ goods out of her $2|A_2|+2$ approved goods, which violates EF1.
    
        Formally, from (\ref{eqn:A1>Ak+1}), we have
        \begin{equation} \label{eqn:case2_A1A2}
            |A_1| \geq |A_2| + 2.
        \end{equation}
        By minimal completeness and IDU, we may assume that
        \begin{enumerate}[(i)]
            \item for all $i \in \{1,2\}$ and $j \in N \setminus \{1,2\}$, $P_i \cap A_j = \emptyset$ and $P_j \cap A_i = \emptyset$;
            \item $|P_2 \cap A_1| = 1$; and
            \item $P_1 = A_1$.
        \end{enumerate}
        The valuations and bundles of agents $\{1,2\}$ in the critical path of $(\instance, \allocation)$ are illustrated as follows.
        \begin{center}
        \begin{tabular}{ c | c c c c c c c c c}
        $\profile$ & \multicolumn{4}{c}{$\overbrace{\smash{}\hspace{3cm}\smash{}}^{A_1 = P_1}$} & \multicolumn{3}{c}{$\overbrace{\smash{}\hspace{2.1cm}\smash{}}^{A_2}$} \\
         \hline
         $1$ & \LL{$1$} & $\cdots$ & $1$ & \RR{$1$}\\
         \cline{2-8}
         $2$ & & & & $1$ & \LL{$1$} & $\cdots$ & \RR{$1$}\\  
         \cline{6-8}
         \multicolumn{1}{c}{} & & & & \multicolumn{1}{c}{$\underbrace{\smash{}\hspace{0.2cm}\smash{}}_{P_2 \cap A_1}$}
        \end{tabular}
        \end{center}
        
        For $\alpha\in \{0,\dots,|A_2| + 1\}$, let $S_\alpha$ be some set of $\alpha$ goods in $A_1 \setminus P_2$; such a set exists by (ii) and (\ref{eqn:case2_A1A2}).       
        Consider any instance $\instance^\alpha = (N, G, \profile^\alpha)$ with profile $\profile^\alpha = (P^\alpha_1,P^\alpha_2)$ such that $P_1^\alpha = P_1$ and $P_2^\alpha = P_2 \cup S_\alpha$. 
        Let $\allocation^\alpha = \therule(\instance^\alpha)$ be the corresponding allocation returned by $\therule$ on the instance $\instance^\alpha$.
        The valuations of both agents under $\instance^\alpha$ are illustrated as follows. We also include the allocation $\allocation$ and the set $S_\alpha$; note that $\allocation$ does \emph{not} denote the allocation returned by $\therule$ in this instance.
        \begin{center}
        \begin{tabular}{ c | c c c c c c c c c c c c}
        $\profile^\alpha$
         & \multicolumn{7}{c}{$\overbrace{\smash{}\hspace{5.3cm}\smash{}}^{A_1}$} & \multicolumn{3}{c}{$\overbrace{\smash{}\hspace{1.8cm}\smash{}}^{A_2}$} \\
         \hline
         $1$ & \LL{$1$} & $\cdots$ & $1$ & $1$ & $\cdots$ & $1$ & \RR{$1$}  \\
         \cline{2-8} \cline{9-11}
         $2$ & & &  & $\mathbf{1}$ & $\cdots$ & $\mathbf{1}$ & $1$ & \LL{$1$} & $\cdots$ & \RR{$1$}\\  
         \cline{9-11}
         \multicolumn{4}{c}{} & \multicolumn{3}{c}{$\underbrace{\smash{}\hspace{1.8cm}\smash{}}_{S_\alpha}$} & \multicolumn{1}{c}{$\underbrace{\smash{}\hspace{0.1cm}\smash{}}_{P_2 \cap A_1}$}
        \end{tabular}
        \end{center}

        Since agent~$2$ is the only agent who values goods in $A_2$, by PO, we have $A_2 \subseteq A^\alpha_2$.
        We will prove by induction that $|A^\alpha_2| =  |A_2|$ for all $\alpha\in[|A_2|+1]$.
        The base case $\alpha = 0$ holds trivially. 
        We establish the induction step via the following lemma.

        \begin{lemma}
        \label{lem:induction}
            Suppose that for any instance $\instance^\kappa = (N, G, \profile^\kappa)$ with $S_\kappa$ for some $\kappa\in\{0, 1, \dots, |A_2|\}$, we have $|A^\kappa_2| = |A_2|$. Then, for any instance $\instance^{\kappa+1} = (N, G, \profile^{\kappa+1})$ with $S_{\kappa+1}$, we also have $|A^{\kappa+1}_2| = |A_2|$.
        \end{lemma}
        
        \begin{innerproof}[Proof of \Cref{lem:induction}]
        Assume that $|A^\kappa_2| = |A_2|$ for any instance $\instance^\kappa$ with $S_\kappa$.
        Consider the instance $\instance^{\kappa+1} = (N, G, \profile^{\kappa+1})$ such that $P^{\kappa+1}_1 = P^\kappa_1$ and $P^{\kappa+1}_2 = P^\kappa_2 \cup \{g\}$ for some $g \in A^\kappa_1 \setminus P^\kappa_2$. 
        Such a good $g$ exists since $|A^\kappa_1| \ge |A_2|+2$ by the induction hypothesis and $A^\kappa_1 \subseteq P_1$, while $|P^\kappa_2\cap P_1|\le |A_2|+1$.

        Let $\allocation^{\kappa+1}= \therule(\instance^{\kappa+1})$ be the corresponding allocation returned by $\therule$ on the instance $\instance^{\kappa+1}$.
        If $g^* \in A^{\kappa+1}_2$ for some good $g^* \in P^\kappa_2 \cap A^\kappa_1$, then agent~$2$ can misreport her preferences from $\profile^\kappa$ to $\profile^{\kappa+1}$ so as to obtain a strictly higher value, contradicting the strategyproofness of $\therule$.
        Thus, $g^* \notin A^{\kappa+1}_2$ for all $g^* \in P^\kappa_2 \cap A^\kappa_1$.

        If $g \notin A^{\kappa+1}_2$, then by the previous paragraph, the induction hypothesis, and PO, we have $A^{\kappa+1}_1 = A^\kappa_1$.
        This implies that $|A^{\kappa+1}_2| = |A^\kappa_2| = |A_2|$, where the latter equality follows from the induction hypothesis, as desired.
        Suppose now that $g \in A^{\kappa+1}_2$.
        Consider another instance $\widehat{\instance} = (N, G, \widehat{\profile})$ with $\widehat{\profile}$ such that $\widehat{P}_1 = P^{\kappa+1}_1 = P^\kappa_1$ and $\widehat{P}_2 = P^{\kappa+1}_2 \setminus \{g^*\}$ for some $g^* \in P^\kappa_2 \cap A^\kappa_1$.
        Let $\widehat{\allocation} = \therule(\widehat{\instance})$ be the corresponding allocation returned by $\therule$ on the instance $\widehat{\instance}$.
        By IDU, $\widehat{\allocation} = \allocation^{\kappa+1}$.
        In particular, we have 
        \begin{equation} \label{eqn:k=2_2}
            |\widehat{A}_2| = |A^{\kappa+1}_2|.
        \end{equation}
        Note that $\widehat{\profile} = \pi(\profile^\kappa)$ for some permutation $\pi$ of $G$.
        By neutrality, $|\widehat{A}_2| = |A^\kappa_2| = |A_2|$, where the latter equality follows from the induction hypothesis.
        Combining this with~(\ref{eqn:k=2_2}), we get $|A^{\kappa+1}_2| = |\widehat{A}_2| = |A_2|$, as desired.
        \end{innerproof}
        Thus, by induction, we showed that for any instance $\instance^\alpha$ where $\alpha = |A_2|+1$, we have $|A^\alpha_2| = |A_2|$.
        However, observe that $|P^\alpha_2| = 2(|A_2|+1)$ while $|A^\alpha_2| = |A_2|$, so EF1 is violated.
        This yields the desired contradiction.
    \end{proof}
To complement \Cref{thm:main-characterization}, we demonstrate that all five properties are necessary for the characterization---in other words, if we drop any of them, there exists an allocation rule that satisfies the remaining four properties but does not maximize Nash welfare.
The proof of this statement can be found in Appendix~\ref{app:independence-main}.
\begin{proposition}\label{prop:non-uniqueness_4of5}
    Under binary valuations, given only four of the properties in \{EF1, strategyproofness, neutrality, minimal completeness, IDU\}, there exists an allocation rule for $n = 2$ that satisfies the four properties but does not maximize Nash welfare.
\end{proposition}

Now, consider any instance with $n = 2$ agents and $m$ goods. 
Assume that $m = 2r+1$ is odd and, for each agent, it is possible to obtain an MNW allocation that gives $r$ goods to that agent and $r+1$ goods to the other agent.
A natural question is whether an allocation rule satisfying the five properties in \Cref{thm:main-characterization} needs consistent tie-breaking, i.e., must consistently break ties in favor of a particular agent across all such instances.
For example, consider the following two instances:
        \begin{center}
        \begin{tabular}{ c | c c c c c }
            $\profile_1$ & $g_1$ & $g_2$ & $g_3$ & $g_4$ & $g_5$\\ 
             \hline
             $1$ & \circled{$1$} & \circled{$1$} & \circled{$1$} & $1$ & $1$ \\  
             $2$ & $1$ & $1$ & $1$ & \circled{$1$} & \circled{$1$} \\
        \end{tabular}
            \quad
            and
            \quad
        \begin{tabular}{ c | c c c c c }
            $\profile_2$ & $g_1$ & $g_2$ & $g_3$ & $g_4$ & $g_5$\\ 
             \hline
             $1$ & \circled{$1$} & \circled{ $1$} & $1$ & $1$ & $1$ \\  
             $2$ & $0$ & $1$ & \circled{$1$} & \circled{$1$} & \circled{$1$} \\
        \end{tabular}
        \end{center}
In both instances, an MNW allocation can give either three goods to agent~$1$ and two goods to agent~$2$, or two goods to agent~$1$ and three goods to agent~$2$.
Thus, a rule that always maximizes Nash welfare may return the corresponding circled allocation in each of the two instances.
However, we will show that this cannot be the case for a rule that satisfies the five properties in \Cref{thm:main-characterization}.
Specifically, we show that among all instances with the same number of total goods and the same number of valued goods, any such rule must break ties consistently.

Let $N = \{1,2\}$. 
For an allocation rule $\therule$, we say that an instance $\instance = (N, G, \profile)$ is \emph{tiebreak-relevant} if for $\therule(\instance) = \allocation$ where $\mathcal{A} \in$ \mnw($\instance$), either
\begin{enumerate}[(i)]
    \item $|A_1| = r+1$, $|A_2| = r$, and $A_1 \cap P_2 \neq \emptyset$, or
    \item  $|A_1| = r$, $|A_2| = r+1$, and $A_2 \cap P_1 \neq \emptyset$,
\end{enumerate}
for some $r \in \mathbb{Z}_{\geq 0}$.
If (i) holds, we say that $\therule$ favors agent~$1$, while if (ii) holds, we say that $\therule$ favors agent~$2$.
Our formal result is stated below.

\begin{theorem} \label{thm:tiebreaking_1stchar}
    Let $n = 2$, and let $\mathcal F$ be a rule that is minimally complete, neutral, IDU, and maximizes Nash welfare.
    Then, for each $m$ and each $m_v\in\{1,\dots,m\}$, there exists an agent~$i \in \{1,2\}$ such that for every tiebreak-relevant instance with $m$ goods, of which $m_v$ are valued, $\therule$ chooses an MNW allocation favoring agent~$i$.
\end{theorem}

\begin{proof}
    Consider any tiebreak-relevant instance $\mathcal{I} = (N,G,\profile)$, and  assume without loss of generality that $\therule(\instance) = \allocation$ is an MNW allocation favoring agent~$1$.
    Suppose for contradiction that there exists some other tiebreak-relevant instance $\instance' = (N,G,\profile')$ with the same number of goods $m$ and valued goods $m_v \in \{1,\dots, m\}$ such that $\therule(\instance')=\allocation'$ is an MNW allocation favoring agent~$2$.
    In both instances, there is the same number of unvalued goods, which are discarded by the rule (by minimal completeness). 
    Let the number of valued goods be $m_v = 2r+1$.
    Then, we have that
    \begin{equation*}
        |A_1|=r+1, \quad |A_2|=r, \quad |A'_1|=r, \quad |A'_2|=r+1.
    \end{equation*}

    Consider the instance $\widehat{\instance} = (N,G,\widehat{\profile})$ with profile $\widehat{\profile} = (\widehat{P}_1,\widehat{P}_2)$ such that
    \begin{itemize}
        \item $\widehat{P}_1 = A_1$, and
        \item $\widehat{P}_2 = A_2 \cup \{ \widehat{g} \}$ for some $\widehat{g} \in A_1 \cap P_2$.
    \end{itemize}
    The valuations of agents $N=\{1,2\}$ under this instance are illustrated as follows. 
    Empty cells indicate that the agent has a value of $0$ for the good. Unvalued goods are not allocated and therefore omitted.
    \begin{center}
        \begin{tabular}{ c | c c c c c c c }
    \multirow{2}{*}{ $\widehat{\profile}$
        }
         & \multicolumn{4}{c}{$\overbrace{\smash{}\hspace{2.4cm}\smash{}}^{A_1}$} & \multicolumn{3}{c}{$\overbrace{\smash{}\hspace{1.8cm}\smash{}}^{A_2}$}\\
         & & & & \rotatebox{0}{$\widehat{g}$} & & &\\ 
         \hline
         $1$ &  $1$ & $\cdots$ & $1$ & $1$ &  & & \\
         $2$ & & & & $1$ & $1$ & $\cdots$ & $1$ \\  
        \end{tabular}
    \end{center}
    Let $\widehat{\allocation} = \therule(\widehat{\profile})$ be the allocation returned by $\therule$ on profile $\widehat{\profile}$.
    Since $\therule$ is IDU, we have that $\widehat{\allocation} = \allocation$.
    In particular,
    \begin{equation} \label{tiebreak_contradict1}
        |\widehat{A}_1| = r+1 \quad \text{and} \quad |\widehat{A}_2| = r.
    \end{equation}
    
    Similarly, consider the instance $\widehat{\instance}' = (N,G,\widehat{\profile}')$ with profile $\widehat{\profile}' = (\widehat{P}'_1,\dots,\widehat{P}'_n)$ such that
    \begin{itemize}
        \item $\widehat{P}'_1 = A'_1 \cup \{\widehat{g}'\} \text{ for some } \widehat{g}' \in A'_2 \cap P'_1$, and
        \item $\widehat{P}'_2 = A'_2$.
    \end{itemize}
    The valuations of agents $N=\{1,2\}$ under this instance are illustrated as follows. 
    \begin{center}
        \begin{tabular}{ c | c c c c c c c}
    \multirow{2}{*}{ $\widehat{\profile}'$
        }
         & \multicolumn{3}{c}{$\overbrace{\smash{}\hspace{1.8cm}\smash{}}^{A'_1}$} & \multicolumn{4}{c}{$\overbrace{\smash{}\hspace{2.4cm}\smash{}}^{A'_2}$}\\
         & & & &  \rotatebox{0}{$\widehat{g}'$} & & &\\ 
         \hline
         $1$ & $1$ & $\cdots$ & $1$ & $1$  &  &  &  \\
         $2$ & & & & $1$ & $1$ & $\cdots$ & $1$ \\  
        \end{tabular}
    \end{center}
    Let $\widehat{\allocation}' = \therule(\widehat{\profile}')$ be the allocation returned by $\therule$ on profile $\widehat{\profile}'$.
    Since $\therule$ is IDU, we have that $\widehat{\allocation}' = \allocation'$.
    In particular,
    \begin{equation} \label{tiebreak_contradict2}
        |\widehat{A}'_1| = r \quad \text{and} \quad |\widehat{A}'_2| = r+1.
    \end{equation}
   However, since the number of unvalued goods is the same in both instances, neutrality implies that $|\widehat{A}'_k| = |\widehat{A}_k|$ for all $k \in \{1,2\}$, a contradiction to (\ref{tiebreak_contradict1}) and (\ref{tiebreak_contradict2}).
\end{proof}

The tie-breaking specification in \Cref{thm:tiebreaking_1stchar} is already the best possible that we can infer from the assumptions of the theorem.
To see this, consider a version of MNW that, when given $m$ total goods and $m_v$ valued goods, discards the unvalued goods and, if necessary, breaks ties\footnote{Note that tie-breaking between agents is never necessary when $m_v$ is even.} in favor of agent $i_{m,m_v}\in\{1,2\}$.
Any remaining ties are broken lexicographically over the goods.\footnote{See the end of \Cref{sec:prelims} for further explanation of this step.}
Regardless of how the indices $i_{m,m_v}$ are chosen across all pairs $(m,m_v)$, the resulting rule satisfies all the assumptions of \Cref{thm:tiebreaking_1stchar}.
Indeed, minimal completeness and the fact that the rule maximizes Nash welfare are obvious.
For neutrality, permuting the labels of the goods does not change $m$ and $m_v$, so the tie-breaking between agents (if any) remains consistent.
Finally, for IDU, if an agent ceases to approve a good that is not allocated to her, then in the original instance, this good must be allocated to the other agent, who also values it.
Therefore, the original instance and the new instance share the same $m$ and $m_v$, and any tie-breaking again remains consistent.

\section{Alternative Characterization for Two Agents}
\label{sec:alternative-characterization}

In this section, we provide an alternative characterization for the setting of two agents, using non-redundancy along with another property called resource-monotonicity in place of IDU. 
Resource-monotonicity states that when an extra good is added, no agent's value should decrease.
\begin{definition}[Resource-monotonicity]
    An allocation rule $\mathcal{F}$ is \emph{resource-monotone} if the following holds: For any two instances $\mathcal{I}$ and $\mathcal{I}'$ such that $\mathcal{I}'$ can be obtained from $\mathcal{I}$ by adding one extra good, if $\mathcal{F}(\mathcal{I}) = \mathcal{A}$ and $\mathcal{F}(\mathcal{I}') = \mathcal{A}'$, then for each $i \in N$, $v_i(A_i) \leq v'_i(A'_i)$.
\end{definition}
\citet{SuksompongTe22} showed that the \mnwtie{} rule satisfies resource-monotonicity.

For $\{i,j\} = \{1,2\}$ and a non-redundant bundle $A_i$, let $A_i = A_i^u \cup A_i^b$ be the partition of $A_i$ into the set of goods that agent~$i$ \emph{uniquely} values (i.e., $A_i^u = A_i \setminus P_j$) and the set of goods that \emph{both} agents value (i.e., $A_i^b = A_i \cap P_j)$, respectively.
Our characterization result is as follows.

\begin{theorem} \label{thm:n=2}
    For $n=2$, under binary valuations, any allocation rule $\therule$ that is \emph{EF1}, non-redundant, strategyproof, resource-monotone, neutral, and minimally complete maximizes Nash welfare.
\end{theorem}

\begin{proof}
    Let $N = \{1,2\}$, and let $\therule$ be a rule satisfying the six axioms in the theorem statement.
    We show that it maximizes Nash welfare by induction on the number of goods $m = |G|$.
    By non-redundancy, for an allocation $\allocation$ returned by $\therule$ in any instance, we have $v_i(A_i) = |A_i|$ for all $i \in N$.
    
    For the base case $m=1$, if the good is unvalued, the rule leaves it unallocated (by minimal completeness).
    If the good is valued, the rule allocates it to some agent who values it (by minimal completeness and non-redundancy).
    Hence, the rule maximizes Nash welfare.

    We introduce some additional notation.
    For every instance $\instance = (N, G, \profile)$ with $n=2$ agents, we define the \emph{characteristic tuple} $\mathcal{C}$ of the corresponding profile $\profile$, which returns a set of goods and three numbers:
    \begin{equation*}
        \mathcal{C}(\profile) = (G, |P_1 \cap P_2|, |P_1 \setminus P_2|, |P_2 \setminus P_1|).
    \end{equation*}
    Then, for any other instance $\instance' = (N, G, \profile')$ with corresponding profile $\profile'$ that has the same characteristic tuple as $\profile$, we have that $\profile' = \pi(\profile)$ for some permutation $\pi$ of $G$.
    Consequently, by neutrality, if $\therule(\instance) = \allocation$ and $\therule(\instance') = \allocation'$, then $v_i(A_i) = v'_i(A'_i)$ for each $i \in N$.

    We prove the induction step through the following lemma.

\begin{lemma}
\label{lem:MNW-induction}
    For any $\kappa\in \mathbb{N}$, if $\therule$ maximizes Nash welfare for every instance with $\kappa$ goods, then $\therule$ also maximizes Nash welfare for every instance with $\kappa+1$ goods.
\end{lemma}

\begin{innerproof}[Proof of \Cref{lem:MNW-induction}]
    Assume that $|G| = \kappa$ and consider any instance $\instance^{\kappa+1} = (N, G\cup\{g^*\}, \profile^{\kappa+1})$ on $\kappa + 1$ goods.
    Let $\instance^\kappa = (N, G, \profile^\kappa)$ be the instance with $\kappa$ goods such that $P^\kappa_i = P^{\kappa+1}_i \cap G$ for all $i\in N$.
    Let $\allocation = \therule(\instance^\kappa)$ be the allocation returned by $\therule$ on the instance $\instance^\kappa$.
    By assumption, $\mathcal{A} \in $ \mnw($\instance^\kappa$). 
    Let $\mathcal{B} = \therule(\instance^{\kappa+1})$ be the allocation returned by $\therule$ on the instance~$\instance^{\kappa+1}$.
    The characteristic tuple of $\profile^{\kappa+1}$ is
    \begin{equation*}
         \mathcal{C}(\profile^{\kappa+1}) = (G \cup \{g^*\}, |B^b_1| + |B^b_2|, |B^u_1|, |B^u_2|).
    \end{equation*}
    The valuations of both agents in $\instance^{\kappa+1}$ are illustrated as follows, together with the allocation $\mathcal{B}$.
    \begin{center}
        \begin{tabular}{ c | c c c c c c c c c c c c}
        $\profile^{\kappa+1}$
         & \multicolumn{3}{c}{$\overbrace{\smash{}\hspace{1.8cm}\smash{}}^{B_1^u}$} & \multicolumn{3}{c}{$\overbrace{\smash{}\hspace{1.8cm}\smash{}}^{B_1^b}$} & 
         \multicolumn{3}{c}{$\overbrace{\smash{}\hspace{1.8cm}\smash{}}^{B_2^b}$} & 
         \multicolumn{3}{c}{$\overbrace{\smash{}\hspace{1.8cm}\smash{}}^{B_2^u}$}\\
         \hline
         $1$ & \LL{$1$} & $\cdots$ & $1$& $1$ & $\cdots$ & \RR{$1$} & $1$ & $\cdots$ & $1$  \\
         \cline{2-13}
         $2$ & & & & $1$ & $\cdots$ & $1$ & \LL{$1$} & $\cdots$ & $1$ & $1$ & $\cdots$ & \RR{$1$} \\  
         \cline{8-13}
        \end{tabular}
    \end{center}
    
    If $g^* \notin P^{\kappa+1}_1 \cup P^{\kappa+1}_2$ (i.e., $g^*$ is unvalued), then by minimal completeness, $g^*$ must remain unallocated.
    By resource-monotonicity, $|B_1| = |A_1|$ and $|B_2| = |A_2|$, and $\therule$ continues to maximize Nash welfare for $\instance^{\kappa+1}$.

    Assume that $g^* \in P^{\kappa+1}_1 \cup P^{\kappa+1}_2$, and suppose for contradiction that $\therule$ does not maximize Nash welfare for the instance $\mathcal{I}^{\kappa+1}$, i.e., $\mathcal{B} \notin $ \mnw($\instance^{\kappa+1}$).
    We say that $\allocation = (A_1,A_2)$ is \emph{balanced} if $||A_1| - |A_2|| \leq 1$, and \emph{unbalanced} otherwise.    
    We split our analysis into two cases according to whether $\allocation$ is balanced. 

        \paragraph{Case 1: $\mathcal{A}$ is unbalanced.}
        This means that $||A_1|-|A_2|| \geq 2$.
        Without loss of generality, assume that
        \begin{equation}\label{eqn:A2A1geq2}
            |A_2| \geq |A_1| + 2.
        \end{equation}
        If $A^b_2 \neq \emptyset$, then $(2,1)$ is a critical path in $(\instance^\kappa, \allocation)$.
        Moreover, non-redundancy and minimal completeness imply PO, so $\allocation$ is PO.
        By Lemma \ref{lem:mnw_conditions}, $\allocation \notin$ \mnw($\instance^\kappa$), contradicting our assumption.
        We may thus assume that $A^b_2 = \emptyset$.
        This means that
        \begin{equation}\label{eqn:A20_nonempty}
            |A^u_2| = |A_2| \geq 2.
        \end{equation}

        If $B^b_2 = \emptyset$, then $\mathcal{B} \in $ \mnw($\instance^{\kappa+1}$), contradicting our assumption.
        Thus, we assume that $B^b_2 \neq \emptyset$. 
        By resource-monotonicity and minimal completeness, this means that $|B_2^b| = 1$, $B_2^u = A_2^u$, $|B_1| = |A_1|$, and $|B_2| = |A_2| + 1$.

        Consider another instance $\widehat{\mathcal{I}} = (N, G \cup \{g^*\} \setminus \{\widehat{g}\}, \widehat{\profile})$ derived from $\mathcal{I}^{\kappa+1}$ by removing some good $\widehat{g} \in B_2^u$. 
        Note that $\widehat{P}_1 = P^{\kappa+1}_1$ and $\widehat{P}_2 = P^{\kappa+1}_2 \setminus \{\widehat{g}\}$.     
        By resource-monotonicity, exactly one agent's bundle value should decrease (by $1$) when going from $\mathcal{I}^{\kappa+1}$ to $\widehat{\mathcal{I}}$.
        Let $\widehat{\mathcal{B}} = \therule(\widehat{\instance})$ be the allocation returned by $\therule$ on the instance $\widehat{\instance}$.
        The valuations of both agents in $\widehat{\instance}$ are illustrated as follows, together with the goods identified with respect to the allocation $\mathcal{B}$.
        Note that $\mathcal{B}$ does \emph{not} denote the allocation returned by $\therule$ in this instance.
        \begin{center}
        \begin{tabular}{ c | c c c c c c c c c c c c}
        $\widehat{\profile}$
         & \multicolumn{3}{c}{$\overbrace{\smash{}\hspace{1.8cm}\smash{}}^{B_1^u}$} & \multicolumn{3}{c}{$\overbrace{\smash{}\hspace{1.8cm}\smash{}}^{B_1^b}$} & 
         \multicolumn{3}{c}{$\overbrace{\smash{}\hspace{1.8cm}\smash{}}^{B_2^b}$} & 
         \multicolumn{3}{c}{$\overbrace{\smash{}\hspace{1.8cm}\smash{}}^{B_2^u \setminus \{\widehat{g}\}}$}\\
         \hline
         $1$ & $1$ & $\cdots$ & $1$& $1$ & $\cdots$ & $1$ & $1$ & $\cdots$ & $1$  \\
         $2$ & & & & $1$ & $\cdots$ & $1$ & $1$ & $\cdots$ & $1$ & $1$ & $\cdots$ & $1$ \\  
        \end{tabular}
    \end{center}
        
        Recall from (\ref{eqn:A2A1geq2}) and (\ref{eqn:A20_nonempty}) that
        \begin{equation*}
            |B_2^u| - 1 = |A_2^u| - 1 = |A_2| - 1 \geq |A_1| + 1.
        \end{equation*}      
        Since $\therule$ maximizes Nash welfare in all instances with $\kappa$ goods (in particular, $\widehat{\instance}$), we have that $\widehat{\mathcal{B}} \in $ \mnw($\widehat{\instance}$).

        Note that $|\widehat{P}_2| = |P_2^{\kappa+1}| - 1 \ge |B_2^u| - 1 \geq |A_1| + 1$.
        We now show that $|\widehat{P}_1| \geq |A_1|+1$ as well.
        Suppose for contradiction that $|\widehat{P}_1| \leq |A_1|$. This means that $|P^{\kappa+1}_1| = |\widehat{P}_1| \leq |A_1| = |B_1|$.
        Hence, agent~$1$ receives from $\mathcal{B}$ all the goods that she values in $\mathcal{I}^{\kappa+1}$.
        It follows that $\mathcal{B} \in$ \mnw$(\instance^{\kappa+1})$, contradicting our assumption.
        Hence, $|\widehat{P}_1| \geq |A_1|+1$.

        The total number of valued goods in instance $\widehat{\mathcal{I}}$ is
        \begin{equation*}
            |B_1| + |B_2| - 1
            = |A_1| + (|B_2^b| + |B_2^u|) -  1 
            \ge |A_1| + 1 + (|B_2^u| - 1)
            \geq 2 (|A_1| + 1).
        \end{equation*}
        Together with the fact that 
        \begin{equation*}
            |\widehat{P}_1| \geq |A_1| + 1 \quad \text{and} \quad |\widehat{P}_2| \geq |A_1|+1,
        \end{equation*}
        in order to ensure that $\widehat{\mathcal{B}}$ maximizes Nash welfare, we must have that $|\widehat{B}_1| \geq |A_1| + 1$.
        
        Finally, recall that instance $\instance^{\kappa+1}$ can be derived from $\widehat{\mathcal{I}}$ by adding the good $\widehat{g}$, which is only valued by agent~$2$.
        By resource-monotonicity, agent~$1$'s bundle value cannot decrease when moving from $\widehat{\mathcal{I}}$ to $\instance^{\kappa+1}$.
        This means that $|B_1| \geq |\widehat{B}_1| \geq |A_1| + 1$.
        However, this contradicts the fact that $|B_1| = |A_1|$.

    \paragraph{Case 2: $\mathcal{A}$ is balanced.}
        Assume that all goods are valued; the proof proceeds similarly if there are unvalued goods.
        If $\kappa$ is even, then $|A_1| = |A_2|$. Since $\therule$ is resource-monotone and one agent's value must increase by exactly $1$ when moving from $\instance^\kappa$ to $\instance^{\kappa+1}$, we have that $\mathcal{B} \in $ \mnw($\instance^{\kappa+1}$), contradicting our assumption.
        Thus, we assume that $\kappa$ is odd.

        Let $\kappa+1 = 2\ell$, for some $\ell \in \mathbb{N}$.
        Without loss of generality, assume that $|A_1| = \ell-1$ and $|A_2| = \ell$.
        If $\mathcal{B}$ is balanced, then  $\mathcal{B} \in$ \mnw($\instance^{\kappa+1}$), contradicting our assumption.
        Thus, it must be that $\mathcal{B}$ is unbalanced. By resource-monotonicity, this means that $|B_1| = |A_1| = \ell-1$ and $|B_2| = |A_2| + 1 = \ell+1$.
        Note also that
        \begin{equation} \label{eqn:B10B20_leqk}
            |B_2^u| \leq \ell;
        \end{equation}
        otherwise $\mathcal{B} \in $ \mnw($\instance^{\kappa+1}$), contradicting our assumption.
        
        Consider the instance $\mathcal{I}^* = (N, G \cup \{g^*\}, \profile^*)$ derived from $\instance^{\kappa+1}$ by letting all goods in $B_2^u$ be valued by agent~$1$---that is, $P^*_1 = P^{\kappa+1}_1 \cup B_2^u = B_1 \cup B_2$ and $P^*_2 = P^{\kappa+1}_2$.
        Let $\mathcal{B}^* = \therule(\mathcal{I}^*)$ be the corresponding allocation returned by $\therule$ on the instance $\mathcal{I}^*$.
        The characteristic tuple of $\profile^*$ is
        \begin{equation*}
            \mathcal{C}(\profile^*) = (G \cup \{g^*\}, |B^b_1| + |B_2|, |B_1^u|, 0).
        \end{equation*}
        The valuations of both agents in $\instance^*$ are illustrated as follows, together with the goods identified with respect to the allocation $\mathcal{B}$.
        Note that $\mathcal{B}$ does \emph{not} denote the allocation returned by $\therule$ in this instance.

        \begin{center}
        \begin{tabular}{ c | c c c c c c c c c c c c}
        $\profile^*$
         & \multicolumn{3}{c}{$\overbrace{\smash{}\hspace{1.8cm}\smash{}}^{B_1^u}$} & \multicolumn{3}{c}{$\overbrace{\smash{}\hspace{1.8cm}\smash{}}^{B_1^b}$} & 
         \multicolumn{3}{c}{$\overbrace{\smash{}\hspace{1.8cm}\smash{}}^{B_2^b}$} & 
         \multicolumn{3}{c}{$\overbrace{\smash{}\hspace{1.8cm}\smash{}}^{B_2^u}$}\\
         \hline
         $1$ & $1$ & $\cdots$ & $1$& $1$ & $\cdots$ & $1$ & $1$ & $\cdots$ & $1$ & $\mathbf{1}$ & $\cdots$ & $\mathbf{1}$  \\
         $2$ & & & & $1$ & $\cdots$ & $1$ & $1$ & $\cdots$ & $1$ & $1$ & $\cdots$ & $1$ \\  
        \end{tabular}
    \end{center}
        
        We have $\ell > |B_1| \ge |B_1^u|$.
        Also, by EF1, it holds that $|B_1^*| \ge \ell$.
        Let $S\subseteq G \cup \{g^*\}$ be a subset of $\ell$ goods that contains
        \begin{enumerate}[(i)]
            \item all $|B_1^u|$ goods from $B_1^u$, and 
            \item $\ell - |B_1^u|$ goods from $B_1^* \setminus B_1^u$. 
        \end{enumerate}
        Now, consider another instance $\widetilde{\mathcal{I}} = (N, G \cup \{g^*\}, \widetilde{\profile})$ constructed from $\mathcal{I}^*$ by declaring some $|B_2^u|$ goods in $B_1^b \cup B_2$ to be unvalued by agent~$1$.
        We choose any such $|B_2^u|$ goods from the set $(B_1^b \cup B_2) \setminus S$.
        Note that $|B_1 \cup B_2| = 2\ell$, and so $|(B_1^b\cup B_2) \setminus S| = |(B_1 \cup B_2) \setminus S| = 2\ell - \ell = \ell$.
        Since $|B_2^u| \leq \ell$ by \eqref{eqn:B10B20_leqk}, such $|B_2^u|$ goods in $(B_1^b \cup B_2) \setminus S$ exist.
        Let $\widetilde{\mathcal{B}} = \therule(\widetilde{\instance})$ be the corresponding allocation returned by $\therule$ on the instance $\widetilde{\instance}$.

        The characteristic tuple of $\widetilde{\profile}$ is
        \begin{equation*}
            \mathcal{C}(\widetilde{\profile}) = (G \cup \{g^*\}, |B^b_1| + |B^b_2|, |B_1^u|, |B_2^u|).
        \end{equation*}
        The valuations of both agents in $\widetilde{\instance}$ are illustrated as follows, together with the goods identified with respect to the bundles $B_1^u$ and $B^b_1 \cup B_2$ and the set $S$.
        Note that $\mathcal{B}$ does \emph{not} denote the allocation returned by $\therule$ in this instance, and the order of goods in $B_1^b \cup B_2$ may differ from that illustrated in $\profile^*$ above.

        \begin{center}
        \begin{tabular}{ c | c c c c c c c c c c c c c c}
        $\widetilde{\profile}$
         & \multicolumn{3}{c}{$\overbrace{\smash{}\hspace{1.8cm}\smash{}}^{B_1^u}$} & 
         \multicolumn{9}{c}{$\overbrace{\smash{}\hspace{6.2cm}\smash{}}^{B_1^b \cup B_2}$}\\
         \hline
         $1$ & $1$ & $\cdots$ & $1$& $1$ & $\cdots$ & $1$ & $1$ & $\cdots$ & $1$ & $\mathbf{0}$ & $\cdots$ & $\mathbf{0}$  \\
         $2$ & & & & $1$ & $\cdots$ & $1$ & $1$ & $\cdots$ & $1$ & $1$ & $\cdots$ & $1$ \\ 
         \multicolumn{1}{c}{} & \multicolumn{6}{c}{$\underbrace{\smash{}\hspace{4cm}\smash{}}_{S}$}  & & & & \multicolumn{3}{c}{$\underbrace{\smash{}\hspace{1.8cm}\smash{}}_{|B_2^u| \text{ goods}}$}
        \end{tabular}
        \end{center}
        
        Since $\mathcal{C}(\widetilde{\profile}) = \mathcal{C}(\profile^{\kappa+1})$, by neutrality, we must have that
        \begin{equation*}
            |\widetilde{B}_1| = |B_1| = \ell-1 \quad \text{and} \quad |\widetilde{B}_2| = |B_2| = \ell+1.
        \end{equation*}
        Moreover, recall that $|B^*_1| \geq \ell$.
        Comparing instance $\widetilde{\mathcal{I}}$ to $\mathcal{I}^*$, we observe that the rule is not strategyproof---agent~$1$ can misreport her valuations and obtain a strictly higher value, giving us a contradiction. 
\end{innerproof}
\Cref{lem:MNW-induction} together with induction implies \Cref{thm:n=2}.
\end{proof}

The proof of \Cref{thm:n=2} does not readily generalize to $n > 2$. 
Indeed, assuming neutrality and minimal completeness, a profile for two agents (with binary valuations) can be described by three numbers. 
Our proof modifies instances while keeping these three numbers fixed, then invokes strategyproofness and resource‑monotonicity to derive a contradiction. 
With $n \geq 3$ agents, $2^n - 1$ numbers are required to describe a profile, which makes it difficult to construct a similar proof.

Next, we present a result showing that any rule satisfying five of the six properties (we do not require non-redundancy here) and maximizing Nash welfare must consistently break ties in tiebreak-relevant instances with the same number of valued goods.
As an immediate corollary, the same must hold for any rule satisfying all six properties, by Theorem \ref{thm:n=2}.
Note that because of resource-monotonicity, this tie-breaking is irrespective of the total number of goods.
For example, given a tiebreak-relevant instance, if we add an unvalued good to it, resource-monotonicity implies that both agents' utilities---and therefore the tie-breaking---must remain the same.
The tie-breaking requirement here is more restrictive than in \Cref{thm:tiebreaking_1stchar}, where consistent tie-breaking is only demanded for instances with the same number of goods (in addition to having the same number of valued goods).

\begin{theorem} \label{thm:tie-breaking_alternate}
    Let $n = 2$, and let $\mathcal F$ be a rule that is \emph{EF1}, minimally complete, resource-monotone, strategyproof, neutral, and maximizes Nash welfare.
    Then, for each $m_v$, there exists an agent~$i \in \{1,2\}$ such that for every tiebreak-relevant instance with $m_v$ valued goods,  $\therule$ chooses an MNW allocation favoring agent~$i$.
\end{theorem}

\begin{proof}
    Due to minimal completeness and resource-monotonicity, each agent's bundle value should remain the same upon adding or removing unvalued goods.
    Thus, without loss of generality, we can assume that all goods are valued in every instance that we consider.

    Suppose for contradiction that there exists a tiebreak-relevant instance with $m$ valued goods such that $\therule$ returns an MNW allocation favoring agent~$1$, and another instance such that $\therule$ returns an MNW allocation favoring agent~$2$. 
    As tie-breaking is never relevant for an even number of valued goods, we may let $m = 2k+1$.
    Consider an instance $\mathcal{I}^*$ with $m$ valued goods, where all goods are valued by both agents, and let $\therule(\instance^*) =\mathcal{A}^*$.
    Since $\mathcal{A}^*$ is EF1, one agent must receive a value of $k+1$ and the other agent a value of $k$.
    Without loss of generality, assume that
    \begin{equation*}
        |A^*_1| = k \quad \text{and} \quad |A^*_2| = k+1.
    \end{equation*}

    Let $\instance$ be a tiebreak-relevant instance with $m$ valued goods such that $\therule(\instance) = \allocation$ is an MNW allocation favoring agent~$1$.
    We have  
    \begin{equation*}
        |A_1| = k+1 \quad \text{and} \quad |A_2| = k.
    \end{equation*}
    Since this instance is tiebreak-relevant, $|A_1^b| \geq 1$ (the notation $A_1^b$ is defined before \Cref{thm:n=2}).
    This means that
    \begin{equation} \label{eqn:Ai0Aj'0_k}
        |A^u_1| \leq k.
    \end{equation}
    
    Next, consider the modified instance $\widehat{\mathcal{I}} = (N, G, \widehat{\profile})$ derived from $\mathcal{I}$ by letting agent~$1$ value all goods in $A^u_2$, i.e., $\widehat{P}_1 = A_1 \cup A_2$ and $\widehat{P}_2 = P_2$.
    Let $\widehat{\mathcal{A}} = \mathcal{F}(\widehat{\mathcal{I}})$.
    If $|\widehat{A}_1| = k$, agent~$1$ can strictly benefit by misreporting her valuations from $\widehat{\profile}$ to~$\profile$, violating strategyproofness.
    Hence,
    \begin{equation*}
        |\widehat{A}_1| = k+1 \quad \text{and} \quad |\widehat{A}_2| = k.
    \end{equation*}

    Now, consider the instance derived from $\mathcal{I}^*$ by letting agent~$2$ have value $0$ for some $|A_1^u|$ goods in $A^*_1$; the existence of such goods is guaranteed by (\ref{eqn:Ai0Aj'0_k}).
    Then, by neutrality with $\widehat{\mathcal{I}}$, agent~$1$ should get a value of $k+1$ and agent~$2$ a value of $k$.
    However, comparing this instance to $\mathcal{I}^*$, agent~$2$ can obtain a strictly higher value (with respect to the goods that she values in this instance) by misreporting her valuations.
    This violates strategyproofness, a contradiction.
\end{proof}

Finally, we prove that all six properties are necessary for our characterization.
In particular, if we were to drop any one of EF1, non-redundancy, strategyproofness, neutrality, or minimal completeness, the allocation rule may no longer maximize Nash welfare. 
On the other hand, if we were to drop resource-monotonicity, then the allocation rule may fail to break ties consistently as mandated by \Cref{{thm:tie-breaking_alternate}}.

\begin{proposition}\label{prop:n=2_alt_char_four}
    Under binary valuations, given only five of the properties in \{EF1, non-redundancy, strategyproofness, resource-monotonicity, neutrality, minimal completeness\}, there exists an allocation rule for $n=2$ that satisfies the five properties but either does not maximize Nash welfare or fails the tie-breaking requirement set out in \Cref{thm:tie-breaking_alternate}.
\end{proposition}

The proof of \Cref{prop:n=2_alt_char_four} can be found in Appendix~\ref{app:independence-two}.

\section{Conclusion and Future Work}

In this work, we studied the setting of fair division with binary valuations and presented characterizations of the rule that coincides with the maximum Nash welfare (MNW) rule, the leximin rule, as well as any additive welfarist rule with a strictly concave function.
Our main result shows that for any number of agents, this rule is the unique rule satisfying EF1, strategyproofness, neutrality, minimal completeness, and IDU.
We also provide an alternative characterization for the case of two agents, by replacing IDU with non-redundancy and resource-monotonicity.

Besides extending \Cref{thm:n=2} to $n > 2$ or obtaining characterizations using other sets of axioms, a natural future direction is to characterize MNW or the leximin rule among all rules for additive valuations.
As discussed in \Cref{sec:intro}, this appears to be a challenging task---for example, several properties fulfilled by MNW in the binary setting are violated by the rule in the additive setting.
More broadly, while characterizations are commonplace in the social choice literature, they are still few and far between in fair division despite the recent surge of interest in the area.
We believe that characterizing other prominent fair division rules---such as the classic \emph{round-robin algorithm} or the influential \emph{envy cycle elimination algorithm} \citep{lipton2004ece}---is a direction worthy of further investigation.

\section*{Acknowledgments}
 This work was partially supported by the Singapore Ministry of Education under grant number MOE-T2EP20221-0001, by the DFG under the Excellence Strategy EXC-2047, and by an NUS Start-up Grant.
 We would like to thank the associate editor and reviewers for their valuable comments.

\bibliographystyle{plainnat}
\bibliography{main}

\appendix
\section*{APPENDIX}
\section{Proof of \Cref{prop:non-uniqueness_4of5}}
\label{app:independence-main}

We show that for each of the five properties (EF1, minimal completeness, neutrality, IDU, strategyproofness), there exists an allocation rule for $n = 2$ that does not maximize Nash welfare and fails only that property among the five.
We include an example demonstrating the violation of each property (and MNW) by the proposed rule, and show that the rule satisfies the remaining four properties.
\begin{itemize}
    \item \textbf{EF1}: A rule that discards all unvalued goods, allocates goods valued by both agents to agent~$1$, and allocates each of the remaining goods to the agent who uniquely values the good. This is also equivalent to the utilitarian rule (with lexicographic tie-breaking).
    \begin{tcolorbox}[breakable]
        To see that the rule is not EF1, consider an instance with $m =2$ goods that are valued by both agents.
        Then, the rule will allocate both goods to agent~$1$, violating EF1.
        The same example shows that the rule returns an allocation that is not MNW.
    \end{tcolorbox}
    
    This rule satisfies:
    \begin{itemize}
        \item \textbf{Minimal completeness.} Every valued good is allocated to an agent who values it, and every unvalued good is discarded.
        \item \textbf{Neutrality.} It is clear that the agents' received values do not depend on the labeling of the goods.
        \item \textbf{IDU.} Suppose an agent ceases to approve a good~$g$ that she did not receive.
        Then, this agent must be agent~$2$, and $g$ is approved by (and allocated to) agent~$1$.
        Upon agent~$2$ disapproving~$g$, the allocation that the rule returns does not change.
        \item \textbf{Strategyproofness.} Agent~$1$ already receives all the goods that she approves, so she cannot increase her value by misreporting. 
        Agent~$2$ receives exactly the goods that she uniquely approves.
        No matter how agent~$2$ reports, she cannot get the goods that agent~$1$ values.
        Therefore, agent~$2$ also cannot increase her value by misreporting.
    \end{itemize}
    \item \textbf{Minimal completeness}: The \mnwtie{} rule favoring agent~$1$, with one exception: if $m=1$, then always discard the good.
    \begin{tcolorbox}[breakable]
        To see that the rule is not minimally complete, observe that it does not allocate a valued good when $m=1$ and both agents value the good.
        This allocation is also not MNW.
    \end{tcolorbox}
    We know that the \mnwtie{} rule favoring agent~$1$ satisfies all the properties (and in particular, the four other properties). Therefore, it remains to show that for instances where $m=1$, the rule that always discards the good satisfies:
    \begin{itemize}
        \item \textbf{EF1.} The returned allocation is $(\varnothing,\varnothing)$, which satisfies EF1 trivially.
        \item \textbf{Neutrality.} With one good, the labeling is unique, and neutrality trivially holds.
        \item \textbf{IDU.} Since the rule always discards the good, the allocation is the same regardless of the agents' reports, and IDU trivially holds.
        \item \textbf{Strategyproofness.} Since the rule always discards the good, both agents always receive a value of~$0$.
    \end{itemize}
            
    \item \textbf{Neutrality}: The \mnwtie{} rule favoring agent~$1$, with the exception that when $m=2$ and the profile is either of the following $\profile_1$ or $\profile_2$, it returns the indicated allocation.
    \begin{center}
        \begin{tabular}{ c | c c }
            $\profile_1$ & $g_1$ & $g_2$\\ 
             \hline
             $1$ & \circled{$1$} & \circled{$1$} \\  
             $2$ & $0$  & $1$ \\
            \end{tabular}
            \quad
            and
            \quad
            \begin{tabular}{ c | c c }
             $\profile_2$ & $g_1$ & $g_2$\\ 
             \hline
             $1$ & $1$ & \circled{$1$} \\  
             $2$ & \circled{$1$}  & $1$ \\
            \end{tabular}
        \end{center}
        \begin{tcolorbox}[breakable]
            To see that the rule is not neutral, consider the following profile:
        \begin{center}
            \begin{tabular}{ c | c c }
            $\profile$ & $g_1$ & $g_2$\\ 
             \hline
             $1$ & $1$ & \circled{$1$} \\  
             $2$ & \circled{$1$} & $0$ \\
            \end{tabular}
        \end{center}
            Then, the \mnwtie{} rule favoring agent~$1$ will return the above indicated allocation, giving both agents $1$ and $2$ a value of $1$ each.
            However, $\profile = \pi(\profile_1)$ for some permutation $\pi$ of $G$, and by neutrality the value that each agent gets in both instances should be equal, which is not the case. The allocation returned by the rule for $\profile_1$ is not MNW.
        \end{tcolorbox}
        We know that the \mnwtie{} rule favoring agent~$1$ satisfies all the properties (and in particular, the four other properties).
        Therefore, it remains to show that for instances where the profile is either $\mathcal{P}_1$ or $\mathcal{P}_2$ (as defined above), the rule that returns the corresponding indicated allocation satisfies:
        \begin{itemize}
            \item \textbf{EF1.}
            Under $\mathcal{P}_1$, $v_2(A_1) = v_2(g_2) = 1$, and so $v_2(A_1 \setminus \{g_2\}) = 0$. Thus, agent~$2$ does not envy agent~$1$ after removing one good, and EF1 is satisfied. Under $\mathcal{P}_2$, each agent receives exactly one good, and EF1 is again satisfied.
            \item \textbf{Minimal completeness.}
            In both profiles, both goods are valued, and both are also allocated.
            \item \textbf{IDU.}
            Under $\mathcal{P}_1$, if agent~$2$ ceases to approve $g_2$, then \mnwtie{} favoring agent~$1$ will return the same allocation (i.e., both goods allocated to agent~$1$).
            Similarly, under $\mathcal{P}_2$, if agent~$1$ ceases to approve $g_1$ or agent~$2$ ceases to approve $g_2$, then \mnwtie{} will return the same allocation. 
            \item \textbf{Strategyproofness.} 
            Under $\mathcal{P}_1$, agent~$1$ is already allocated both goods, so she has no incentive to misreport.
            Agent~$2$ has no way to misreport so that she receives $g_2$; in particular, if she reports that she values both goods, then the profile becomes $\mathcal{P}_2$, and agent~$2$ receives only $g_1$.
            Also, if agent~$1$ actually approves only one good, then the rule already gives her that good, so she has no incentive to report that she values both goods so that the profile becomes $\mathcal{P}_1$.
            Now, consider $\mathcal{P}_2$.
            Each agent already receives a value of~$1$, and there is no way either agent can misreport in order to obtain both goods.
            Also, if either agent actually approves one good, then the rule already gives her that good (except in $\mathcal{P}_1$, which we already handled), so the agent has no incentive to report that she values both goods so that the profile becomes $\mathcal{P}_2$.
        \end{itemize}
        
        \item \textbf{IDU}: The \mnwtie{} rule favoring agent~$2$, with the exception that when $m=1$ and only agent~$1$ values the good, allocate it to agent~$2$.
        \begin{tcolorbox}[breakable]
            To see that the rule is not IDU, consider the profile where $m=1$ and the good is unvalued.
            Then, \mnwtie{} favoring agent~$2$ will discard the good.
            However, IDU states that the allocation should remain the same as when only agent~$1$ values the good, and thus the rule does not satisfy it. The allocation returned by the rule when only agent~$1$ values the good is not MNW.
        \end{tcolorbox}
        We know that the \mnwtie{} rule favoring agent~$2$ satisfies all the properties (and in particular, the four other properties).
        Therefore, it remains to show that for the instance where $m=1$ and only agent~$1$ values the good, the rule that allocates the good to agent~$2$ satisfies:
    \begin{itemize}
        \item \textbf{EF1.} 
        Agent~$1$'s envy is eliminated by removing the single good in agent~$2$'s bundle, so EF1 is satisfied.
        \item \textbf{Minimal completeness.} 
        There is only one good, which is valued and allocated.
        \item \textbf{Neutrality.}
        With one good, the labeling is unique, and neutrality trivially holds.
        \item \textbf{Strategyproofness.}   
        Agent~$2$ already receives the good and has no incentive to misreport.
        If agent~$1$ misreports her value for the good (i.e., $0$ rather than~$1$), then \mnwtie{} favoring agent~$2$ will discard the good, and agent~$1$'s value remains unchanged.
        On the other hand, if agent~$2$ actually approves the good, then \mnwtie{} favoring agent~$2$ already gives her the good, so she has no incentive to misreport.
        Likewise, if agent~$1$ actually disapproves the good, she has no incentive to misreport.
    \end{itemize}
        \item \textbf{Strategyproofness}:
        The \mnwtie{} rule favoring agent~$2$, with one exception: when $m=5$, $|P_1| \geq 4$, $|P_2| = 2$, and all goods are valued, return any allocation $(A_1, A_2)$ with $|A_1| = v_1(A_1) = 4$ and $|A_2| = v_2(A_2) = 1$.
        \begin{tcolorbox}[breakable]
        To see that the rule is not strategyproof, consider the following two profiles:
        \begin{center}
        \begin{tabular}{ c | c c c c c }
            $\profile_1$ & $g_1$ & $g_2$ & $g_3$ & $g_4$ & $g_5$\\ 
             \hline
             $1$ & $0$ &$0$ & \circled{$1$} & $1$ & $0$ \\  
             $2$ & $0$ & $0$ & $0$ & \circled{$1$} & \circled{$1$} \\
        \end{tabular}
            \quad
            and
            \quad
        \begin{tabular}{ c | c c c c c }
            $\profile_2$ & $g_1$ & $g_2$ & $g_3$ & $g_4$ & $g_5$\\ 
             \hline
             $1$ & \circled{$\mathbf{1}$} & \circled{ $\mathbf{1}$} & \circled{$1$} & \circled{$1$} & $0$ \\  
             $2$ & $0$ & $0$ & $0$ & $1$ & \circled{$1$} \\
        \end{tabular}
        \end{center}
        For $\profile_1$, the \mnwtie{} rule favoring agent~$2$ will return the above indicated allocation, giving agent~$1$ a value of $1$.
        For $\profile_2$, the rule will return the indicated allocation, thereby giving agent~$1$ a value of $2$ with respect to her true preference $\profile_1$. This allocation is also not MNW.
        \end{tcolorbox}
        We know that the \mnwtie{} rule favoring agent~$2$ satisfies all the properties (and in particular, the four other properties).
        Therefore, it remains to show that for instances where $m=5$, $|P_1| \geq 4$, $|P_2| = 2$, and all goods are valued, the rule as described above satisfies:
    \begin{itemize}
        \item \textbf{EF1.} 
        Agent~$2$ values two of the five goods, and she receives one of these two goods. 
        Agent~$1$ already receives four valued goods.
        Thus, EF1 is satisfied.
        \item \textbf{Minimal completeness.} 
        All goods are valued and allocated.
        \item \textbf{Neutrality.} 
        By description of the rule, the agents' received values do not depend on the labeling of the goods.
        \item \textbf{IDU.}
        If agent~$2$ ceases to approve the (single) good that she values but did not receive, the rule will return the same allocation.
        For agent~$1$, when $|P_1| = 4$, the agent already receives all of her valued goods; when $|P_1| = 5$, if the agent ceases to approve the (single) good that she values but did not receive, the rule will return the same allocation.
        Also, if agent~$2$ actually approves at least three goods, then \mnwtie{} favoring agent~$2$ will give her at least three valued goods, so even if she ceases to approve a good that she values but did not get, in the new profile she will value at least three goods.
    \end{itemize}
\end{itemize}

\section{Proof of \Cref{prop:n=2_alt_char_four}}
\label{app:independence-two}

We show that for each of the six properties (EF1, non-redundancy, minimal completeness, neutrality, resource-monotonicity, strategyproofness), there exists an allocation rule for $n = 2$ that does not maximize Nash welfare (or in the case of resource-monotonicity, fails the tie-breaking requirement set out in \Cref{thm:tie-breaking_alternate}) and fails only that property among the six.
We include an example demonstrating the violation of each property (and MNW) by the proposed rule.
\begin{itemize}
    \item \textbf{EF1}: We use the same rule that violates EF1 in the proof of Proposition~\ref{prop:non-uniqueness_4of5}. 
    We have already shown that this rule satisfies minimal completeness, neutrality, and strategyproofness, but does not maximize Nash welfare. We now show that it satisfies:
    \begin{itemize}
        \item \textbf{Non-redundancy.}
        Every allocated good is given to an agent who values it.
        \item \textbf{Resource-monotonicity.}
        When an extra good is added, the original goods are allocated in the same way as before, so no agent's value decreases.
    \end{itemize}
    \item \textbf{Non-redundancy}:
    We use the same rule that violates IDU in the proof of Proposition~\ref{prop:non-uniqueness_4of5}.
    \begin{tcolorbox}[breakable]
        To see that the rule fails non-redundancy, observe that when $m=1$ and the good is only valued by agent~$1$, agent~$2$ receives a good that she does not value. The allocation returned by the rule in this case is not MNW.
    \end{tcolorbox}
    We have already shown that this rule satisfies EF1, minimal completeness, neutrality, and strategyproofness, but does not maximize Nash welfare.
    Therefore, it remains to show that for the instance where $m=1$ and only agent~$1$ values the good, the rule that allocates the good to agent~$2$ satisfies:
    \begin{itemize}
        \item \textbf{Resource-monotonicity.}
        Since both agents have a value of $0$, adding an extra good cannot reduce either agent’s value. Thus, resource-monotonicity holds.
    \end{itemize}
    \item \textbf{Minimal completeness}: The \mnwtie{} rule favoring agent~$1$, with one exception: if $m=1$ and both agents value the good, discard it.
    \begin{tcolorbox}[breakable]
        To see that the rule is not minimally complete, observe that it does not allocate a valued good when $m=1$ and both agents value the good.
        This allocation is also not MNW.
    \end{tcolorbox}
    We know that the \mnwtie{} rule favoring agent~$1$ satisfies all the properties (and in particular, the five other properties).
    Therefore, it remains to show that for the instance where $m=1$ and both agents value the good, the rule that discards the good satisfies:
    \begin{itemize}
        \item \textbf{EF1.} The returned allocation is $(\varnothing, \varnothing)$, which satisfies EF1 trivially.
        \item \textbf{Non-redundancy.}
        Since each agent's bundle is empty, non-redundancy trivially holds.
        \item \textbf{Neutrality.}
        With one good, the labeling is unique, and neutrality trivially holds.
        \item \textbf{Resource-monotonicity.}
        Since both agents have a value of $0$, adding an extra good cannot reduce either agent’s value. Thus, resource-monotonicity holds.
        \item \textbf{Strategyproofness.}
        If either agent misreports her value for the good (i.e., $0$ instead of $1$), then \mnwtie{} will allocate the good to the other agent. 
        Also, if only one agent actually approves the good, the agent who does not approve the good cannot gain by misreporting.
        
    \end{itemize}
    \item \textbf{Neutrality}: The \mnwtie{} rule favoring agent~$1$, with the exception that for all instances with two valued goods (denote these goods by $g_i$ and $g_j$, where $i < j$) and any number of unvalued goods:
    \begin{enumerate}[(i)]
        \item if both agents value $g_i$ and $g_j$, then allocate $g_i$ to agent~$2$ and $g_j$ to agent~$1$;
        \item if agent~$1$ values both $g_i$ and $g_j$, whereas agent~$2$ values only $g_j$, then allocate both $g_i$ and $g_j$ to agent~$1$.
    \end{enumerate}
    In both cases, the unvalued goods are discarded.
    \begin{tcolorbox}[breakable]
        To see that the rule is not neutral, consider the following three profiles (with no unvalued goods):
        \begin{center}
        \begin{tabular}{ c | c c }
            $\profile_1$ & $g_1$ & $g_2$\\ 
             \hline
             $1$ & \circled{$1$} & \circled{$1$} \\  
             $2$ & $0$  & $1$ \\
            \end{tabular}
            \quad
            and
            \quad
            \begin{tabular}{ c | c c }
             $\profile_2$ & $g_1$ & $g_2$\\ 
             \hline
             $1$ & $1$ & \circled{$1$} \\  
             $2$ & \circled{$1$}  & $1$ \\
            \end{tabular}
            \quad
            and
            \quad
            \begin{tabular}{ c | c c }
            $\profile_3$ & $g_1$ & $g_2$\\ 
             \hline
             $1$ & $1$ & \circled{$1$} \\  
             $2$ & \circled{$1$} & $0$ \\
            \end{tabular}
        \end{center}
            The rule will return the above indicated allocations, using (ii) for $\profile_1$, (i) for $\profile_2$, and the \mnwtie{} rule favoring agent~$1$ for $\profile_3$.
            However, $\profile_3 = \pi(\profile_1)$ for some permutation $\pi$ of $G$, so by neutrality the value that each agent gets in both of these instances  should be equal, which is not the case. The allocation returned by the rule for $\profile_1$ is not MNW.
    \end{tcolorbox}
    We know that the \mnwtie{} rule favoring agent~$1$ satisfies all the properties (and in particular, the five other properties).
    Therefore, it remains to show that for instances fulfilling (i) or (ii), the rule as described above satisfies:
    \begin{itemize}
        \item \textbf{EF1.}
        In case (i), each agent receives a value of $1$, whereas in case (ii), agent~$2$'s envy towards agent~$1$ disappears after removing $g_j$ from agent~$1$'s bundle.
        \item \textbf{Non-redundancy.}
        Every valued good is allocated to an agent who values it.
        \item \textbf{Minimal completeness.}
        Every valued good is allocated, and every unvalued good is discarded.
        \item \textbf{Resource-monotonicity.}
        Suppose an extra good $g$ is added.
        If $g$ is unvalued, then the rule will discard it and the allocation of valued goods remains unchanged.
        If $g$ is valued (either by both agents or only by one agent), the rule will return an allocation such that agent~$1$'s value is $2$ and agent~$2$'s value is $1$.
        Hence, no agent's value decreases.
        Also, it can be verified that if a good is removed, no agent's value increases.
        \item \textbf{Strategyproofness.}
        In case (i), each agent has a value of $1$.
        If either agent misreports by disapproving at least one of $g_i$ and $g_j$, and/or approving any subset of the unvalued goods, the rule still will not allocate both $g_i$ and $g_j$ to that agent, so the agent cannot strictly benefit. 
        In case (ii), since agent~$1$ is already allocated both valued goods, she cannot increase her value by misreporting. If agent~$2$ were to misreport by disapproving $g_j$ and/or approving $g_i$ and/or approving some subset of the unvalued goods, the rule still will not allocate $g_j$ to agent~$2$, so she also cannot strictly benefit.
        Moreover, it can be verified that there is no beneficial manipulation to a profile belonging to case (i) or (ii).
    \end{itemize}
    
    \item \textbf{Resource-monotonicity}: The \mnwtie{} rule favoring agent~$1$, with one exception: if $m = 2$ and one good is valued by both agents while the other good is unvalued, allocate the valued good to agent~$2$ and discard the unvalued good.
    \begin{tcolorbox}[breakable]
        To see that the rule is not resource-monotone, consider the instance where $m=1$ and the good is valued by both agents. Then, the \mnwtie{} rule favoring agent~$1$ will allocate this good to agent~$1$.
        When an unvalued good is added, agent~$1$'s bundle becomes empty, violating resource-monotonicity.
        Observe that these two instances have the same number of valued goods.
        However, tie-breaking is in favor of agent~$1$ in the instance with $m=1$ but in favor of agent~$2$ in the instance with $m=2$.
        This violates the tie-breaking requirement set out by \Cref{thm:tie-breaking_alternate}.
    \end{tcolorbox}
    We know that the \mnwtie{} rule favoring agent~$1$ satisfies all the properties (and in particular, the five other properties).
    Therefore, it remains to show that for an instance with $m=2$ where one good is valued by both agents and the other good is unvalued, the rule as described above satisfies:
    \begin{itemize}
        \item \textbf{EF1.}
        Agent~$1$'s envy towards agent~$2$ disappears after removing the single good in agent~$2$'s bundle.
        \item \textbf{Non-redundancy.} Since the valued good is allocated to agent~$2$ (who values it), non-redundancy holds.
        \item \textbf{Minimal completeness.} The valued good is allocated, and the unvalued good is discarded.
        \item \textbf{Neutrality.} It is clear that the agents' received values do not depend on the labeling of the goods.
        \item \textbf{Strategyproofness.}
        Without loss of generality, assume that $g_1$ is valued by both agents and $g_2$ is unvalued.
        If agent~$1$ misreports by disapproving $g_1$ and/or approving $g_2$, her value cannot increase under the rule (since $g_1$ will not be allocated to her in any of those cases).
        Agent~$2$ also cannot strictly benefit by misreporting, since she is already allocated the only good that she values.
        Moreover, it can be verified that there is no beneficial manipulation to this profile.
    \end{itemize}
    
    \item \textbf{Strategyproofness}: The \mnwtie{} rule favoring agent~$1$, with the exception that for all instances with two valued goods (denote these goods by $g_i$ and $g_j$) and any number of unvalued goods, such that agent~$1$ values both $g_i$ and $g_j$ and agent~$2$ values only one of them, allocate both valued goods to agent~$1$ and discard all unvalued goods.
    \begin{tcolorbox}[breakable]
        To see that the rule is not strategyproof, consider the following two profiles (with no unvalued goods):
        \begin{center}
            \begin{tabular}{ c | c c }
             $\profile_1$ & $g_1$ & $g_2$ \\ 
             \hline
             $1$ & \circled{$1$} & \circled{$1$} \\  
             $2$ & $0$ & $1$ \\
            \end{tabular}
            \qquad
            and
            \qquad
            \begin{tabular}{ c | c c }
             $\profile_2$ & $g_1$ & $g_2$ \\ 
             \hline
             $1$ & \circled{$1$} & $1$ \\  
             $2$ & $1$ & \circled{$1$} \\
            \end{tabular}
        \end{center}
        The rule will return the above indicated allocations.
        In $\profile_1$, agent~$2$ can manipulate by reporting $1$ on $g_1$ and increase her value from $0$ to $1$.
        The allocation returned by the rule for $\profile_1$ is not MNW.
    \end{tcolorbox}
    We know that the \mnwtie{} rule favoring agent~$1$ satisfies all the properties (and in particular, the five other properties).
    Therefore, it remains to show that for instances with two valued goods such that agent~$1$ values both of these goods and agent~$2$ values only one of them, the rule as described above satisfies:
    \begin{itemize}
        \item \textbf{EF1.}
        Since agent~$1$ is allocated both valued goods, she cannot envy agent~$2$. Since agent~$2$ only values one good (which is allocated to agent~$1$), by removing that good from agent~$1$'s bundle, any envy from agent~$2$ disappears.
        \item \textbf{Non-redundancy.}
        Both valued goods are allocated to agent~$1$ (who values them).
        \item \textbf{Minimal completeness.}
        Both valued goods are allocated, and all unvalued goods are discarded.
        \item \textbf{Neutrality.}
        It is clear that the agents' received values do not depend on the labeling of the goods.
        \item \textbf{Resource-monotonicity.}
        Suppose an extra good $g$ is added.
        If $g$ is unvalued, then the rule will discard it and the allocation of valued goods remains unchanged.
        If $g$ is valued (either by both agents or only by one agent), the rule will return an allocation such that agent~$1$'s value is $2$ and agent~$2$'s value is $1$.
        Hence, no agent's value decreases.
        Also, it can be verified that if a good is removed, no agent's value increases.
    \end{itemize}
\end{itemize}

\end{document}